\documentclass[
superscriptaddress,
twocolumn,
nofootinbib,
amsmath,amssymb,
aps,
pre,
floatfix,
]{revtex4-2}

\usepackage[switch]{lineno}
\usepackage{graphicx}
\usepackage[dvipsnames]{xcolor}
\usepackage{placeins}
\usepackage[caption=false]{subfig}
\usepackage{epsfig}
\usepackage{float} 
\usepackage{amsthm}
\usepackage{comment}
\usepackage{tikz}
\usetikzlibrary{arrows.meta}
\usepackage{chngcntr}

\newtheorem{proposition}{Proposition}
\newtheorem{corollary}{Corollary}
\newtheorem{remark}{Remark}

\theoremstyle{definition}
\newtheorem{example}{Example}
\newtheorem{definition}{Definition}
\newtheorem*{proposition*}{Proposition}
\newtheorem*{corollary*}{Corollary}

\usepackage{hyperref}
\usepackage{cancel}
\hypersetup{
	unicode=false,          
	pdftoolbar=true,        
	pdfmenubar=true,        
	pdffitwindow=false,     
	pdftitle={},    
	pdfauthor={},     
	pdfsubject={},   
	pdfcreator={},   
	pdfproducer={},  
	pdfkeywords={}, 
	pdfnewwindow=true,      
	colorlinks=false,       
	linkcolor=red,          
	citecolor=green,        
	filecolor=magenta,      
	urlcolor=cyan           
}

\usepackage{cleveref}
\makeatletter
\AddToHook{cmd/appendix/before}{\def\cref@section@alias{appendix}}
\makeatother

\crefformat{figure}{#2Fig.~#1#3}
\crefmultiformat{figure}{#2Figs.~#1#3}{ and~#2#1#3}{, #2#1#3}{ and~#2#1#3}

\crefformat{theorem}{#2Theorem~#1#3}
\crefformat{remark}{#2Remark~#1#3}

\crefformat{enumi}{#2Step~#1#3}
\crefmultiformat{enumi}{#2Steps~#1#3}{ and~#2#1#3}{, #2#1#3}{ and~#2#1#3}

\crefrangeformat{figure}{#3Figs.~#1#4 to~#5#2#6}

\crefformat{appendix}{#2Appendix~#1#3}

\crefformat{corollary}{#2Corollary~#1#3}

\crefformat{definition}{#2Definition~#1#3}
\crefformat{proposition}{#2Proposition~#1#3}

\crefformat{section}{#2Section~#1#3}
\crefmultiformat{section}{#2\S\S~#1#3}{ and~#2#1#3}{, #2#1#3}{ and~#2#1#3}

\crefformat{algorithm}{#2Algorithm~#1#3}
\crefmultiformat{algorithm}{#2Algorithms~#1#3}{ and~#2#1#3}{, #2#1#3}{ and~#2#1#3}

\crefformat{equation}{#2Eq.~(#1)#3}
\crefmultiformat{equation}{#2Eqs.~(#1)#3}{ and~#2(#1)#3}{, #2(#1)#3}{ and~#2(#1)#3}

\crefformat{table}{#2Table~#1#3}
\crefmultiformat{table}{#2Tables~#1#3}{ and~#2#1#3}{, #2#1#3}{ and~#2#1#3}
\def\blue{\textcolor{blue}}

\counterwithin{equation}{section}

\begin{document}
	
	\def\stackalignment{l}
	\begin{abstract}
		Width-balance conditions at a junction are often associated with reflectionless transmission, or transparency, in some one-dimensional wave models on networks. We show that, on cyclic networks, this identification is incomplete: local balance is a vertex-level condition, while transparency also requires synchronization of the path travel times. To make this precise, we consider a shallow-water wave model and derive the channel-width-weighted scattering law for a vertex of arbitrary degree and introduce a source-relative definition of balance, reflecting the fact that the same junction may be reflective or reflectionless depending on which edges carry incoming waves. Under this definition, a balanced vertex is transparent only to the synchronized incoming-amplitude direction. For an \(N\)-path island, the resulting first-generation upstream reflection is governed in frequency space by the width-weighted distribution of path delays. Exact broadband cancellation requires all path travel times to agree; when they do not, commensurability of the path-length differences produces a periodic comb of frequency-selective zeros, independent of channel widths. At \(N\ge 4\) we further exhibit a hybrid regime in which the reflection factor vanishes without commensurability among the path lengths.
	\end{abstract}

	\title{Source-Induced Reflection in Balanced Shallow-Water Networks}
	\author{B. Gu}    
	\affiliation{Department of Mathematical Sciences, Worcester Polytechnic Institute, Worcester, MA 01609, USA.}
	\author{C. Norton}    
	\affiliation{Department of Mathematical Sciences, Worcester Polytechnic Institute, Worcester, MA 01609, USA.}  
	\author{A. Nachbin}    
	\affiliation{Department of Mathematical Sciences, Worcester Polytechnic Institute, Worcester, MA 01609, USA.}
    \affiliation{Instituto de Matemática Pura e Aplicada, Rio de Janeiro, RJ 22460-320, Brazil.}
	
	\pagestyle{plain}
	\renewcommand{\baselinestretch}{1.1}
	\maketitle
	
	\section{Introduction}
	One-dimensional hyperbolic systems on networks describe wave propagation in a wide range of branched physical systems: shallow water in rivers and irrigation canals~\cite{Jacovkis91,LeugeringSchmidt2002,BastinCoron2016,BellamoliMullerToro2018,BrianiPuppoRibot2022}, blood flow in arterial trees~\cite{SherwinIJNF2003,SherwinJEM2003}, gas in pipeline networks~\cite{GugatHerty2011}, and acoustic waves in branched ducts. On each edge of the network, waves propagate along characteristic families; at junctions, algebraic coupling conditions enforce conservation of mass and continuity of an appropriate dynamical variable \cite{Bressan,GaravelloPiccoli2006}. These junction conditions transmit part of the incoming waves and  generally reflect part of them. A central question is therefore when a junction, or a network assembled from junctions, is transparent: an incoming disturbance produces no reflected component in the incoming branches.

    For quantum graphs -- metric graphs equipped with a differential equation along their edges -- transparent behavior is often tied to sum-rule conditions on edge weights, known as balance conditions.  In particular,  degree-three balanced star graphs can reduce to evolution on a line or half-line under suitable symmetry of the incoming data  \cite{Sobirov,PelinovskyGoodman,Yusupov1}. For linear shallow-water networks, Jacovkis~\cite{Jacovkis91} formulated the corresponding coupling conditions for Y-junctions, or degree-three star graphs, using continuity of surface elevation and conservation of width-weighted flux; see \cref{2junctions}.  This motivates the questions addressed in the present work: how does the transparency picture extend to junctions of arbitrary degree, and what additional restrictions appear when locally balanced junctions are assembled into networks with \emph{cycles}? 

	\begin{figure}[!htb]
		\centering
		\includegraphics[width=3.4in,height=2.in]{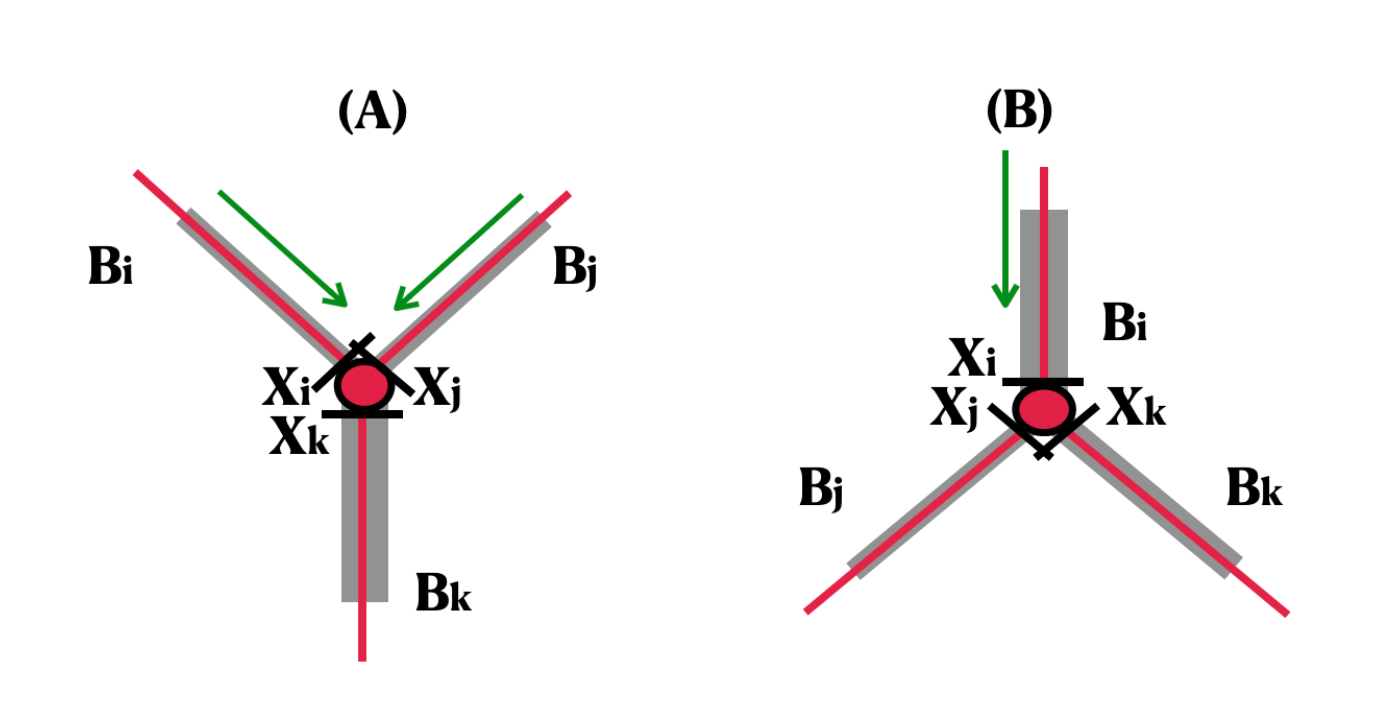}
		\caption{Two Y-shaped degree-three junctions: {\bf (A)} a converging-junction (CJ) and {\bf (B)} a diverging-junction (DJ), all with the respective extreme points $\{ x_i, x_j, x_k\}$. The arrows indicate the directions used for our initial data. The shaded grey areas indicate the different weights that can be imposed over the edges of a network. In the present work, they are channel widths $B_i,B_j$ and $B_k$, respectively.}
		\label{2junctions}
	\end{figure}
	
	Our contributions extend this picture in three ways. First, we derive the channel-width-weighted scattering matrix at a vertex of arbitrary degree. Second, motivated by the fact that the same physical junction can be reflectionless or reflective depending on which edges carry incoming waves, we refine the usual balance condition so that it accounts for the active incoming edges. This yields an algebraic and inherently time-dependent condition on when a balanced junction suppresses reflection. Third, we show that such local balance is not sufficient for transparency on networks with cycles. The island graphs in \cref{Islands} provide a minimal example: local balance may hold at each junction, but unequal path lengths (\cref{fig:asym_island}) can desynchronize the arriving waves and regenerate upstream reflection. We make this observation explicit by deriving the first-generation upstream reflection for two-path and \(N\)-path islands. In frequency space, the \(N\)-path formula yields a reflection factor that separates broadband transparency, which requires equal travel times, from frequency-selective cancellation at the harmonics of a common divisor of the path-length differences.
	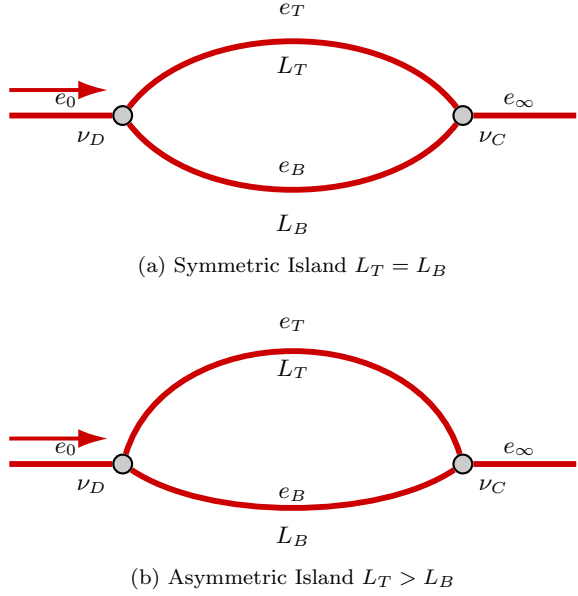
\begin{figure}[htbp]
		\centering
		\subfloat[Symmetric Island $L_T=L_B$]{\label{fig:sym_island}
			\begin{tikzpicture}[
				scale=0.75,
				junction/.style={circle, draw=black, fill=gray!40, thick, inner sep=2.5pt},
				edge/.style={line width=2.2pt, red!80!black},
				flow/.style={-{Latex[length=4mm,width=2.5mm]}, line width=1.5pt, red!80!black},
				lab/.style={font=\small}
				]
				
				\node[junction, label=below left:{$\nu_D$}] (D) at (0,0) {};
				\node[junction, label=below right:{$\nu_C$}] (C) at (6,0) {};
				
				\draw[edge] (-2,0) -- (D);
				\draw[edge] (C) -- (8,0);
				
				\draw[edge] (D) .. controls (1.3,1.7) and (4.7,1.7) .. (C);
				\draw[edge] (D) .. controls (1.3,-1.7) and (4.7,-1.7) .. (C);
				
				\draw[flow] (-2,0.45) -- (-0.25,0.45);
				
				\node[lab, above] at (-1,0) {$e_0$};
				\node[lab, above] at (7,0) {$e_\infty$};
				
				\node[lab, above] at (3,1.6) {$e_T$};
				\node[lab, below] at (3,1.2) {$L_T$};
				
				\node[lab, above] at (3,-1.2) {$e_B$};
				\node[lab, below] at (3,-1.6) {$L_B$};
				
				
			\end{tikzpicture}   
		}
		\\
		\subfloat[Asymmetric Island $L_T>L_B$]{\label{fig:asym_island}
			\begin{tikzpicture}[
				scale=0.75,
				junction/.style={circle, draw=black, fill=gray!40, thick, inner sep=2.5pt},
				edge/.style={line width=2.2pt, red!80!black},
				flow/.style={-{Latex[length=4mm,width=2.5mm]}, line width=1.5pt, red!80!black},
				lab/.style={font=\small}
				]
				
				\node[junction, label=below left:{$\nu_D$}] (D) at (0,0) {};
				\node[junction, label=below right:{$\nu_C$}] (C) at (6,0) {};
				
				\draw[edge] (-2,0) -- (D);
				\draw[edge] (C) -- (8,0);
				
				\draw[edge] (D) .. controls (0.8,2.6) and (5.2,2.6) .. (C);
				\draw[edge] (D) .. controls (1.5,-1.0) and (4.5,-1.0) .. (C);
				
				\draw[flow] (-2,0.45) -- (-0.25,0.45);
				
				\node[lab, above] at (-1,0) {$e_0$};
				\node[lab, above] at (7,0) {$e_\infty$};
				
				\node[lab, above] at (3,2.2) {$e_T$};
				\node[lab, below] at (3,2.0) {$L_T$};
				
				\node[lab, above] at (3,-0.8) {$e_B$};
				\node[lab, below] at (3,-0.95) {$L_B$};
				
				
			\end{tikzpicture}  
		}
		\caption{(a) A symmetric island. (b) An asymmetric island with a longer edge on the top side. This is not a star graph, but has a combination of those in \cref{2junctions}. The disturbance arrives from upstream (the left edge $e_0$).}
		\label{Islands}
	\end{figure}

    Our setting is closest in spirit to recent work on transparent quantum graphs by Yusupov et al.~\cite{Yusupov2,Yusupov2025}, who study the Schr{\"o}dinger evolution on metric star graphs of degree-three under weighted-Kirchhoff continuity. The PDE and the coupling law differ — linearized Saint-Venant with Stoker height-continuity and width-weighted flux here, Schr\"odinger with weighted Kirchhoff there — but the shared question is when a junction transmits without reflection. In most works the star-graphs are of degree-three and have only one leading-edge, namely with the source disturbance. 

    Yusupov et al.~\cite{Yusupov2,Yusupov2025} call attention to the importance of studying transparent graphs.  Having a dominant reflection, rather than transmission, implies a network of large resistivity regarding wave propagation. In applications such as condensed matter, optimal transmission implies ballistic transport of charge, spin, heat and other carriers in low-dimensional branched materials. Important areas, where  ballistic transport in branched structures is required, are for example molecular electronics and conducting polymers ~\cite{Yusupov1}. For water waves, optimal transmission, by contrast, is itself a hazardous situation in which the free surface disturbance can propagate downstream without being scattered. The new ingredient in the present analysis is the synchronization effect created by multiple sources or cycles. These findings are of both physical and mathematical interest.
	
	The paper is organized as follows.  \cref{sec:model} introduces the linearized shallow-water model, Stoker vertex coupling, and the arbitrary-degree scattering law at a vertex. \cref{sec:source_balance} defines source-relative balance and characterizes non-reflection through the source-source block of the scattering matrix. \cref{sec:cycle} applies the theory to the two-path island and \(N\)-path island. \cref{sec:conclusion} summarizes the results and discusses extensions.
	\section{Long-wave dynamics and vertex scattering}
	\label{sec:model}
	Following Jacovkis~\cite{Jacovkis91}, we model long-wave propagation on each channel of the network by the linearized Saint-Venant (shallow-water) equations, 
	\begin{equation}
		\frac{\partial \mathbf{w}}{\partial t} + \mathcal{A}\, \frac{\partial \mathbf{w}}{\partial x} = 0,
		\qquad
		\mathcal{A} = \begin{bmatrix} U & g \\ H & U \end{bmatrix},
		\label{diffeq}
	\end{equation}
	where $\mathbf{w} = (u(x,t), h(x,t))^{\rm T}$, $u$ is the depth-averaged velocity, $h$ is the depth fluctuation about a uniform reference depth $H$, $U$ is the underlying streaming flow, and $g$ is the acceleration due to gravity. The eigenvalues of $\mathcal{A}$ are $\lambda_{1,2} = U \mp \sqrt{gH}$. We restrict throughout this work to the \emph{subcritical regime} $\lambda_1 < 0 < \lambda_2$, equivalently with Froude number $Fr := U/\sqrt{gH} < 1$, in which the long-wave speed $C_0 = \sqrt{gH}$ exceeds the streaming flow and disturbances propagate both upstream and downstream. 
    In order to highlight the reflection-transmission dynamics through symmetric propagating modes, we set \(U=0\) so that
	\begin{equation}
		\lambda_1=-\sqrt{gH},
		\qquad
		\lambda_2=\sqrt{gH}.
		\label{eq:modal_transform}
	\end{equation}
	This symmetric characteristic structure isolates the scattering generated by the junction conditions themselves, which is the central object of the analysis below. 
	
	We diagonalize the edge dynamics, \cref{diffeq}, by writing $\mathbf{y}=(v,z)^{\rm T}:=M^{-1}\mathbf{w}$ where
	\begin{equation}
		M =
		\begin{bmatrix}
			m_{11} & m_{12} \\
			m_{21} & m_{22}
		\end{bmatrix}
		=
		\begin{bmatrix}
			-\sqrt{gH} & \sqrt{gH} \\
			H & H
		\end{bmatrix},
		\label{eq:eigenvector}
	\end{equation} 
	such that $\mathcal{A} = M\Lambda M^{-1}$ where $\Lambda$ is a diagonal matrix with entries $\lambda_{1,2}$. As a result, \cref{diffeq} decouples into two unidirectional wave equations
	\begin{equation}
		\frac{\partial v}{\partial t} + \lambda_1 \frac{\partial v}{\partial x} = 0,
		\qquad
		\frac{\partial z}{\partial t} + \lambda_2 \frac{\partial z}{\partial x} = 0,
		\label{uni-equ}
	\end{equation}
	where \(v\) propagates upstream and \(z\) propagates downstream in the subcritical regime. Moreover, the physical and characteristic variables satisfy the transformation
	\begin{equation}
		u=\sqrt{gH}(z-v),
		\qquad
		h=H(v+z).
		\label{eq:uh_vz_transform}
	\end{equation}
	
	A one-dimensional (1D) fluvial network is represented by a metric graph: edges are 1D-channels and vertices are either open boundary points or junctions. To motivate the vertex  scattering notation, we first consider the two degree-three junctions used by  Jacovkis~\cite{Jacovkis91}: a converging junction (CJ), where two reaches merge into one, and a diverging junction (DJ), where one reach bifurcates into two; see \cref{2junctions}. These Y-junctions are used only as a motivating example. Their role is to show that the same local coupling conditions lead to the same scattering matrix for CJ and DJ configurations, once the incoming and outgoing characteristic modes are identified correctly. This bookkeeping issue is what the arbitrary-degree formulation later resolves.
	
	At each Y-junction the dynamic coupling is enforced by the Stoker compatibility conditions, which coincide with the Neumann--Kirchhoff vertex conditions of quantum-graph  theory~\cite{BerkoKuch} (two formulations we use interchangeably from hereon):
	\begin{subequations}
		\label{simplifiedcc}
		\begin{gather}
			B_i u(x_i, t) \pm B_j u(x_j, t) = B_k u(x_k, t), \label{eq:flux_cons}
			\\
			h(x_i, t) = h(x_j, t) = h(x_k, t),      
		\end{gather}
	\end{subequations}
	where \(B_i,B_j,B_k\) are the channel widths, and \(x_i,x_j,x_k\) are the extremal point evaluated at the junction using the local coordinates of each edge.  The first condition expresses conservation of width-weighted flux and the second expresses continuity of free-surface elevation.  The sign is \(+\) for a CJ and \(-\) for a DJ. Existence and uniqueness of \(C^1\) solutions on such networks under these conditions were established by Jacovkis~\cite{Jacovkis91}; our focus is the scattering generated by the same junction laws. For the similar coupling in a nonlinear setting, see \cite{CaputoDutykh2014} on the 2D-to-1D sine-Gordon reduction at Y- and T-junctions.
	
	In modal variables per \cref{eq:uh_vz_transform}, the Y-junction conditions \cref{simplifiedcc} become
	\begin{subequations}
		\label{ccsub2}
		\begin{gather}
			B_i(m_{11}v_i+m_{12}z_i)
			\pm
			B_j(m_{11}v_j+m_{12}z_j)
			\notag\\
			=
			B_k(m_{11}v_k+m_{12}z_k),
			\label{eq:flux_cons_char}\\
			m_{21}v_i+m_{22}z_i
			=
			m_{21}v_j+m_{22}z_j
			=
			m_{21}v_k+m_{22}z_k.
			\label{eq:cont_char}
		\end{gather}
	\end{subequations}
	Here \(v_i(t)=v(x_i,t)\), \(z_i(t)=z(x_i,t)\), and similarly for \(j,k\).
	
	The CJ and DJ configurations differ only in which characteristic modes are incoming to the vertex and which are outgoing from it.  For a CJ,
	\[
	\mathbf x_{\rm out}^{\rm CJ}=(v_i,v_j,z_k)^T,
	\qquad
	\mathbf x_{\rm in}^{\rm CJ}=(z_i,z_j,v_k)^T,
	\]
	whereas for a DJ,
	\[
	\mathbf x_{\rm out}^{\rm DJ}=(v_i,z_j,z_k)^T,
	\qquad
	\mathbf x_{\rm in}^{\rm DJ}=(z_i,v_j,v_k)^T.
	\]
    Columns vectors are denoted in bold. A direct rearrangement of \cref{eq:flux_cons_char,eq:cont_char} gives, in both cases,
	\[
	\mathbf x_{\rm out}=S(B_i,B_j,B_k)\mathbf x_{\rm in},
	\]
	with
	\begin{align}
		&S(B_i,B_j,B_k)
		=
		\frac{1}{B_i+B_j+B_k} \times \notag\\
		&\begin{bmatrix}
			B_i-B_j-B_k & 2B_j & 2B_k\\
			2B_i & -B_i+B_j-B_k & 2B_k\\
			2B_i & 2B_j & -B_i-B_j+B_k
		\end{bmatrix} \label{eq:scattering_DJCJ}
	\end{align}
	The derivation is given in \cref{app:cjdj}.  Thus the CJ/DJ distinction lies only in the assignment of  the sources, namely the edge-characteristic modes which are incoming or outgoing at the vertex; the scattering law itself is the same.  We now encode this bookkeeping in a vertex-local notation that applies to a vertex of arbitrary degree $d$.
	
	\begin{definition}[Edge orientation at a vertex]\label{def:orientation}
		Following the standard metric-graph convention~\cite{BerkoIntro,BerkoKuch}, each edge \(e\) is identified with an interval \([0,L_e]\) after choosing an arbitrary local
		coordinate \(x_e\). We take \(u_e>0\) to denote flow in the \(+x_e\) direction. The orientation sign of \(e\) at the vertex \(\nu\) is
		\[
		\epsilon_e
		=
		\begin{cases}
			+1, & \nu \text{ is the terminal endpoint } x_e=L_e,\\
			-1, & \nu \text{ is the initial endpoint } x_e=0.
		\end{cases}
		\]
		Equivalently, \(\epsilon_e=+1\) when the \(+x_e\) direction points toward \(\nu\). 
        Usually \(+x_e\) is aligned with the downstream direction, and \(\epsilon_e=+1\) identifies edges feeding \(\nu\) from upstream.
	\end{definition}
    Thus the choice of edge coordinate fixes only a local sign convention. The scattering law derived below does not depend on the arbitrary choice of edge origins; the signs are used to identify which characteristic mode (defined by \cref{uni-equ}) is incoming to, or outgoing from, the vertex. 
    
	\begin{definition}[Vertex input and output modes]
		\label{def:input_output_modes}
		Let \(v_e\) and \(z_e\) denote the upstream- and downstream-propagating characteristic modes on edge \(e\), following \cref{uni-equ}.  At a vertex \(\nu\), 
        record the incoming and outgoing mode-amplitudes through
		\begin{equation}
			\alpha_e
			=
			\begin{cases}
				z_e, & \epsilon_e=+1,\\
				v_e, & \epsilon_e=-1,
			\end{cases}
			\qquad
			\beta_e
			=
			\begin{cases}
				v_e, & \epsilon_e=+1,\\
				z_e, & \epsilon_e=-1.
			\end{cases}
			\label{eq:alpha_beta_def}
		\end{equation}
	\end{definition}
	Thus \(\alpha_e\) records propagation toward \(\nu\), while \(\beta_e\) propagation away from \(\nu\). For \(E_\nu=\{e_1,\dots,e_d\}\), collect these amplitudes into
	\begin{equation}
		\boldsymbol{\alpha}
		=
		(\alpha_{e_1},\dots,\alpha_{e_d})^{\rm T},
		\qquad
		\boldsymbol{\beta}
		=
		(\beta_{e_1},\dots,\beta_{e_d})^{\rm T}.
		\label{eq:xin_xout_def}
	\end{equation}
    Note here that $\boldsymbol{\alpha}$ corresponds to $\mathbf{x}_{\rm in}$ discussed in the special case of CJ and DJ (resp. $\boldsymbol{\beta}$ corresponds to $\mathbf{x}_{\rm out}$). We use the former from hereon.
    
	With the convention given in \cref{def:orientation}, the general Neumann--Kirchhoff  conditions at a vertex are
	\begin{subequations}
		\label{eq:general_NK}
		\begin{align}
			h_e &= h_{e'}, \,\, \forall e, e' \in E_\nu \label{eq:general_continuity}\\
			\sum_{e \in E_\nu} \epsilon_e B_e u_e &= 0. \label{eq:general_flux}
		\end{align}
	\end{subequations}
	We are now in position to derive the scattering matrix for a junction of arbitrary degree directly via the input and output modes. Since the scattering matrix (see \cref{eq:scattering_DJCJ}) is independent of \(g\) and \(H\), we set \(C_0=\sqrt{gH}=1\) and \(H=1\) below for notational simplicity.
	
	\begin{proposition}[Vertex scattering law]
		\label{prop:vertex_scattering}
		Consider a vertex of degree $d$ with incident edge set $E_\nu = \{e_1, \dots, e_d\}$ and positive widths $\{B_e\}_{e \in E_\nu}$. The Neumann--Kirchhoff conditions \cref{eq:general_NK} determine
		\begin{equation}
			\boldsymbol{\beta}=S\boldsymbol{\alpha},
			\label{eq:scattering_law}
		\end{equation}
		where the scattering matrix's entries are
		\begin{equation}
			S_{ab}
			=
			\frac{2B_b}{\sum_{e\in E_\nu}B_e}
			-
			\delta_{ab},
			\qquad a,b\in E_\nu.
			\label{eq:Sgeneraldegree}
		\end{equation}
		Equivalently, in matrix form,
		\begin{equation}
			S
			=
			\frac{2}{\sum_{e\in E_\nu}B_e}\mathbf 1\,\mathbf B^{\rm T}
			-
			I_d,
			\qquad
			\mathbf B=(B_e)_{e\in E_\nu}.
			\label{eq:S_matrix_form}
		\end{equation}
	\end{proposition}
	\begin{proof}
		By \cref{def:input_output_modes}, \(\{\alpha_e,\beta_e\}=\{v_e,z_e\}\) on each edge. Thus, $h_e = v_e + z_e=\alpha_e + \beta_e$, and continuity condition  \eqref{eq:general_continuity} becomes
		\begin{equation}
			\alpha_e + \beta_e = W \quad \text{for all } e \in E_\nu,
			\label{eq:cont_alpha_beta}
		\end{equation}
		where $W$ is the common value of $h_e$ at $\nu$.
		
		To impose condition \cref{eq:general_flux}, define the velocity directed into $\nu$ by $u_e^{\rm in} := \epsilon_e\, u_e$. Direct check from \eqref{eq:alpha_beta_def} and $u_e = z_e - v_e$ gives, in both orientations,
		\begin{equation}
			u_e^{\rm in} = \alpha_e - \beta_e,
			\label{eq:uin_alpha_beta}
		\end{equation}
		again, regardless of orientation. As a result, substituting \cref{eq:uin_alpha_beta} into \cref{eq:general_flux} yields
		\begin{equation}
			\sum_{e \in E_\nu} B_e\,(\alpha_e - \beta_e) = 0.
			\label{eq:flux_alpha_beta}
		\end{equation}
		
		Solving the system \cref{eq:cont_alpha_beta,eq:flux_alpha_beta}:  the first gives $\beta_e = W - \alpha_e$, and substituting into the second yields
		\[
		W = \frac{2\sum_{e\in E_\nu} B_e\,\alpha_e}{\sum_{e\in E_\nu} B_e}.
		\]
		Then for each $a \in E_\nu$,
		\begin{align*}
			\beta_a 
			&= W - \alpha_a \\
			&= \frac{2 \sum_{e\in E_\nu} B_e\, \alpha_e}{\sum_{e\in E_\nu} B_e} - \alpha_a \\
			&= \sum_b \left(\frac{2 B_b}{\sum_{e\in E_\nu} B_e}\right)\, \alpha_b 
			\;-\; \sum_b \delta_{ab}\, \alpha_b \\
			&= \sum_b \!\left(\frac{2 B_b}{\sum_{e\in E_\nu} B_e} - \delta_{ab}\right)\!\alpha_b.
		\end{align*}
		where the term in the parenthesis is precisely $S_{ab}$ given by  \cref{eq:Sgeneraldegree}. The matrix form follows naturally.
	\end{proof}
	It is useful to rewrite the scattering law in a \emph{weighted-mean} form.  For an incoming vector \(\boldsymbol{\alpha}=(\alpha_e)_{e\in E_\nu}\), define
	\[
	\overline{\alpha}_B
	=
	\frac{\sum_{e\in E_\nu}B_e\alpha_e}
	{\sum_{e\in E_\nu}B_e}.
	\]
	Then for each component of \ref{eq:scattering_law} we have
	\begin{equation}
		\beta_e
		=
		2\overline{\alpha}_B-\alpha_e,
		\qquad e\in E_\nu.
		\label{eq:weighted_mean_scattering}
	\end{equation}
	Physically, every outgoing edge sees a common, width-weighted ``junction state'' \(2\overline{\alpha}_B\), with the locally incoming amplitude removed from its own outgoing slot. Reflectionlessness for a given incoming pattern is therefore the question of whether that incoming pattern is a stationary point of this mean-and-subtract operation — the viewpoint to be formalized in \cref{prop:balance_rank}.
	
	In the unweighted case $B_e \equiv 1$, \cref{eq:Sgeneraldegree} reduces to $S_{ab} = 2/d - \delta_{ab}$, which is the bond scattering matrix of a Neumann quantum graph derived by Berkolaiko~\cite{BerkoIntro}, Eq.~(55). 
	
	\section{Source-relative balance}
	\label{sec:source_balance}
	
	We illustrate that on a Y-vertex, the same scattering matrix may produce different reflection behavior under two different choices of incoming data.
	\begin{example}
		Consider a balanced Y-vertex (\cref{2junctions}(A)) with channel widths $B_i=B_j=1/2$, and $B_k=1$. Per \cref{eq:scattering_DJCJ}, the scattering matrix $S$ for this case is 
		\[
			S= 
			\begin{bmatrix}
				-1/2 &1/2 & 1   \\
				1/2& -1/2 & 1 \\
				1/2 & 1/2 & 0  \\
			\end{bmatrix}.
		\]
		For synchronized incoming data \((a(t),a(t),0)^T\), this gives
		\[
        S\left[\begin{array}{c}
        a(t)\\
        a(t)\\
        0
        \end{array}\right]=\left[\begin{array}{c}
        0\\
        0\\
        a(t)
        \end{array}\right]
        \]
		so there is no reflection back into \(i\) or \(j\).  For a single incoming disturbance \((a(t),0,0)^T\),
		\[
        S\left[\begin{array}{c}
        a(t)\\
        0\\
        0
        \end{array}\right]=\left[\begin{array}{c}
        -a(t)/2\\
        a(t)/2\\
        a(t)/2
        \end{array}\right]
        \]
		so the same vertex reflects when only one source edge is active.
	\end{example}
	
	The example shows that reflectionlessness depends on both the widths and the active source designation.  We therefore define a reflectionless vertex relative to a source set. Let \(\nu\) be a vertex with incident edge set \(E_\nu=\{e_1,\dots,e_d\}\), where each edge \(e\in E_\nu\) has positive width \(B_e>0\).
	\begin{definition}[Source set]
		Let \(\alpha_e(t)\) denote the characteristic mode propagating toward
		\(\nu\) along edge \(e\).  At a given scattering time \(t\), the source set
		is
		\[
		\Sigma(t)=\{e\in E_\nu:\alpha_e(t)\neq 0\}.
		\]
		When the active set is fixed during the scattering event under discussion,
		we write it simply as \(\Sigma\).  Edges in \(E_\nu\setminus \Sigma\) are
		quiescent incoming edges for that event.
	\end{definition}
	
	\begin{definition}[Balanced vertex]
		\label{def:balanced_vertex}
		Let $\nu$ be a vertex with incident edge set $E_\nu$, positive widths $\{B_e\}_{e \in E_\nu}$, and source set $\Sigma \subseteq E_\nu$. The vertex $\nu$ is \emph{balanced with respect to $\Sigma$} if
		\begin{equation}
			\sum_{e \in \Sigma} B_e
			\;=\;
			\sum_{e \in E_\nu \setminus \Sigma} B_e.
			\label{eq:balance_def}
		\end{equation}
        The set $\Sigma$ in this definition is taken at a fixed scattering time; we suppress the argument $t$ when no ambiguity arises.  Because the active source set $\Sigma(t)$ may change as new waves arrive at $\nu$, the property ``$\nu$ is balanced with respect to $\Sigma(t)$'' is itself a time-local statement.  The two-path island in \cref{sec:cycle} is the simplest setting in which this temporal dependence has observable consequences.
	\end{definition}
	For a Y-vertex \(E_\nu=\{i,j,k\}\), the balance condition gives the three alternatives
	\[
	B_i=B_j+B_k,\qquad
	B_j=B_i+B_k,\qquad
	B_k=B_i+B_j.
	\]
	Thus the example \(B_i=B_j=1/2\), \(B_k=1\) is balanced with respect to \(\{i,j\}\) (i.e., a CJ), equivalently with respect to \(\{k\}\) (i.e., a DJ), but not with respect to any other source set.
	
	Given a source set \(\Sigma\), decompose the full scattering matrix into blocks according to \(E_\nu=\Sigma\cup(E_\nu\setminus\Sigma)\).  The source-source block \(S_{\Sigma,\Sigma}\) maps incoming amplitudes on \(\Sigma\) to outgoing amplitudes on the same source edges.  Thus \(S_{\Sigma,\Sigma}\alpha_\Sigma=0\) is precisely the no-reflection condition back into the active source set, formalized below.
	
	\begin{proposition}[Balance relative to source set]
		\label{prop:balance_rank}
		The vertex $\nu$ is balanced with respect to $\Sigma$ if and only if the source-source block $S_{\Sigma, \Sigma}$ has a non-trivial kernel, namely
		\begin{equation}
			\sum_{e \in \Sigma} B_e = \sum_{e \in E_\nu \setminus \Sigma} B_e\;\Longleftrightarrow\; \det S_{\Sigma, \Sigma} = 0.
			\label{eq:balance_rank}
		\end{equation}
		At balance, the kernel of $S_{\Sigma, \Sigma}$ contains the constant vector $\mathbf{1}_\Sigma = (1, \dots, 1)^T \in \mathbb{R}^{|\Sigma|}$: synchronized, equal-amplitude incoming data on all edges of $\Sigma$ produce no reflection back into $\Sigma$.
	\end{proposition}
	\begin{proof}
		Write $T = \sum_{e \in E_\nu} B_e$ and $T_\Sigma = \sum_{e \in \Sigma} B_e$. For $a, b \in \Sigma$,
		\[
		(S_{\Sigma, \Sigma})_{ab} = \frac{2 B_b}{T} - \delta_{ab}.
		\]
		Applying $S_{\Sigma, \Sigma}$ to $\mathbf{1}_\Sigma$, we obtain the identity
		\[
		(S_{\Sigma, \Sigma}\, \mathbf{1}_\Sigma)_a
		= \sum_{b \in \Sigma} \left( \frac{2 B_b}{T} - \delta_{ab} \right)
		= \frac{2 T_\Sigma}{T} - 1.
		\]
		For forward direction ``$\Longrightarrow$'', assume $\sum_{e \in \Sigma} B_e = \sum_{e \in E_\nu \setminus \Sigma} B_e$. Equivalently, 
		\[
		T_\Sigma = T - T_\Sigma \iff 2 T_\Sigma = T. 
		\]
		Therefore, $S_{\Sigma,\Sigma}\mathbf{1}_\Sigma=0$, i.e., $S_{\Sigma,\Sigma}$ has a non-trivial kernel. This proves the forward direction ``$\Longrightarrow$''.
		
		For the converse direction ``$\Longleftarrow$'', assume $\det S_{\Sigma, \Sigma} = 0$. Let $m=|\Sigma|$ and write
		\[
		\mathbf{B}_\Sigma=(B_e)_{e\in\Sigma}.
		\]
		From \cref{eq:S_matrix_form}, we find
		\[
		\det S_{\Sigma,\Sigma}
		=
		(-1)^m\left(1-\frac{2T_\Sigma}{T}\right).
		\]
		Thus $\det S_{\Sigma,\Sigma}=0$ implies $2T_\Sigma=T$, which is equivalent to the  condition that $\nu$ is balanced with respect to $\Sigma$. This proves the converse direction ``$\Longleftarrow$''.
	\end{proof}

    Operationally, \cref{prop:balance_rank} says that vertex balance is exactly the condition under which the source-source reflection block $S_{\Sigma,\Sigma}$ annihilates the synchronized direction \(\mathbf 1_\Sigma\); the following remark records the local mechanism we will use repeatedly in \cref{sec:cycle}.
    
	\begin{remark}[Transparent source mode]\label{rem:transparent_source}
		Let \(m=|\Sigma|\).  Since
		\[
		S_{\Sigma,\Sigma}
		=
		-I_m+\frac{2}{T}\mathbf 1_\Sigma \mathbf B_\Sigma^T,
		\qquad
		T=\sum_{e\in E_\nu}B_e,
		\]
		every vector \(q\) satisfying \(\mathbf B_\Sigma^Tq=0\) is an eigenvector with eigenvalue \(-1\), while \(\mathbf 1_\Sigma\) has eigenvalue \(2T_\Sigma/T-1\).  At balance, \(2T_\Sigma=T\), so the transparent source mode is the synchronized direction \(\mathbf 1_\Sigma\).  Thus a balanced multi-source vertex is transparent exactly to the \textbf{synchronized} component on \(\Sigma\). Any departure from  $\mathbf{1}_\Sigma$ projects onto the $(-1)$-eigenspace and is reflected back into $\Sigma$ with sign flipped; this is the local mechanism that generates the  upstream reflection on the cyclic networks analyzed below.
	\end{remark}
    
    A numerical illustration of source-relative balance beyond the Y-vertex example above is given in \cref{fig:deg4} of \cref{sec:numerics}: a degree-4 vertex with a two-edge source set $\Sigma=\{e_1,e_2\}$, satisfying the balance condition per \cref{eq:balance_def}, driven by synchronized incoming pulses on $\Sigma$, generates no reflection back into $\Sigma$ and transmits fully into the complementary edges, in agreement with \cref{rem:transparent_source}.
    
	\section{Balanced islands and first-generation upstream reflection}
	\label{sec:cycle}
	
	The source-relative balance condition characterizes transparency at a single vertex.  We now consider what happens when balanced vertices are combined into a cyclic subgraph.  The simplest example is an island: a diverging junction splits an incoming pulse into several paths, and a converging junction later recombines the partial waves, per \cref{Islands}.  Local balance controls the scattering at each junction, while the path lengths control whether the partial waves arrive at the closing junction in the synchronized source mode.
	
	We focus on the first-generation upstream reflection, denoted by \(\beta_0^{(1)}(t)\): the group of delayed upstream components generated when the initially split waves first recombine at \(\nu_C\) and scatter back to the upstream boundary.  The full reflected signal at the upstream boundary is an infinite sum \(\beta_0(t)=\sum_{n\ge 1}\beta_0^{(n)}(t)\), where \(\beta_0^{(n)}\) collects the delayed amplitudes produced by (n) successive round trips, but only the first generation admits the closed-form formulas below. Because internal measurements in transient-wave diagnostics and control are typically sparse or infeasible (see~\cite{Che2022LoopedFramework,Che2022LoopedRegularization,Che2022ModifiedTreeBoundary,Reddy2011GasLeakEstimator,Mynard2020PressureFlowWaves,DosSantosPrieur2008OpenChannelBoundaryControl,Janon2016OpenChannelSensitivity}), \(\beta_0^{(1)}\) is also the boundary-observable signal: a nonzero \(\beta_0^{(1)}\) is sufficient to rule out transparency, and its amplitude is a tractable boundary-facing measure of the cyclic-subnetwork reflection.
    
	\subsection{Two-path Islands}
	Consider the island graph $G$ in \cref{Islands}: two semi-infinite boundary edges $e_0$ (upstream) and $e_\infty$ (downstream) of widths $B_0, B_\infty$, two finite arc edges $e_T, e_B$ of widths $B_T, B_B$ and lengths $L_T, L_B$, and two interior vertices $\nu_D$ (where $e_0$ meets $e_T, e_B$) and $\nu_C$ (where $e_T, e_B$ meet $e_\infty$). Since (normalized) wave speed $C_0 = 1$ is uniform across the network, the arc lengths $L_T$, $L_B$ also represent travel times.
	
	For an incoming pulse on $e_0$ alone, the boundary source set is $\Sigma = \{e_0\}$. This induces source designations at the interior vertices: $\Sigma_D = \{e_0\}$ at $\nu_D$ and $\Sigma_C = \{e_T, e_B\}$ at $\nu_C$. Vertex-balance (\cref{def:balanced_vertex}) requires $B_T + B_B = B_0$ at $\nu_D$ and $B_T + B_B = B_\infty$ at $\nu_C$, and we therefore impose $B_0 = B_\infty = B_T + B_B$ throughout.
	
	\begin{proposition}[Two-path island synchronization]
		\label{prop:island_sync}
		Suppose vertex-balance holds at both \(\nu_D\) and \(\nu_C\), so that
		\[
		B_0=B_\infty=B_T+B_B.
		\]
		For a nonzero non-periodic incoming pulse on \(e_0\), with amplitude \(a(t)\) at the vertex, the two-path
		island is transparent for all \(t\) if and only if
		\[
		L_T=L_B.
		\]
		If \(L_T\neq L_B\), the first-generation upstream reflection on \(e_0\) is
		\begin{gather}
			\beta_0^{(1)}(t)
			=
			\frac{B_TB_B}{(B_T+B_B)^2}
			\times \notag \\ \left[
			2a(t-L_T-L_B)-a(t-2L_T)-a(t-2L_B)
			\right].
			\label{eq:beta0_two_path}
		\end{gather}
	\end{proposition}
	The proof is deferred to \cref{app:prop3}. From this formula, unequal arc lengths produce a nonzero first-generation upstream reflection for a generic localized pulse, whereas equal arc lengths keep the arriving source vector in the transparent kernel and generate no back-reflection. Numerical visualizations of the synchronized and unsynchronized cases are given in \cref{fig:island_sims} of \cref{sec:numerics}.
	
	\subsection{Multi-path islands}
	\label{subsec:multipath}
	The two-path formula given by \cref{eq:beta0_two_path} suggests a natural extension.  Suppose the island has \(N\) parallel arc edges \(e_1,\dots,e_N\) from \(\nu_D\) to \(\nu_C\), with widths \(B_1,\dots,B_N\) and lengths \(L_1,\dots,L_N\).  Let
	\[
	A=\sum_{k=1}^N B_k.
	\]
	We impose at both junctions:
	\begin{equation}
		B_0=B_\infty=A.
		\label{eq:N_balance}
	\end{equation}
	An incoming pulse \(a(t)\) on \(e_0\) is then split by \(\nu_D\) into identical outgoing characteristic amplitudes on all \(N\) arcs. The two-path formula per \cref{eq:beta0_two_path} generalizes to the following result.
	\begin{proposition}[First-generation upstream reflection for an \(N\)-path island]
		\label{prop:N_path_reflection}
		Under the balance condition \cref{eq:N_balance}, the first-generation upstream reflection on \(e_0\) is
		\begin{gather}
			\beta_0^{(1)}(t)
			=
			\frac{1}{A^2}
			\sum_{1\le k<j\le N}
			B_kB_j \times \notag \\
			\left[
			2a(t-L_k-L_j)-a(t-2L_k)-a(t-2L_j)
			\right].
			\label{eq:beta0_N_time}
		\end{gather}
		Equivalently, with the Fourier transform $\widehat g(\omega):=\int_{\mathbb R} g(t)\,e^{-i\omega t}\,dt$, and writing
        \begin{equation}
        \pi_k=\frac{B_k}{A},
        \qquad
        P(\omega)=\sum_{k=1}^N \pi_k e^{-i\omega L_k},
        \label{eq:frequency_P}
        \end{equation}
        the time-domain identity \cref{eq:beta0_N_time} is equivalent to the pointwise frequency-domain identity
        \begin{equation}
        \widehat{\beta}_0^{(1)}(\omega)
        \;=\;
        \widehat a(\omega)
        \left[
        P(\omega)^2-P(2\omega)
        \right],
        \qquad \omega\in\mathbb R.
        \label{eq:beta0_N_freq}
        \end{equation}
	\end{proposition}
	We defer the proof to \cref{app:prop_4}.  In \cref{sec:numerics}, the three-path case is checked by comparing \cref{eq:beta0_N_time} with the numerically computed first-generation upstream reflection; see \cref{fig:N3_time_domain}.
	
	\begin{corollary}[Cancellation mechanisms for an $N$-path island]
		\label{cor:codim_lengths}
        Let \(\mathcal{H}(\omega):=P(\omega)^2-P(2\omega) \)  where $P(\omega)$ is defined in \cref{eq:frequency_P}. Then the following statements hold.
        \begin{enumerate}
            \item $\mathcal H(\omega)=0 \ \text{for all } \omega\in\mathbb R$ $\Longleftrightarrow$ $L_1=L_2=\cdots=L_N$.
            \item For unequal lengths $\{L_k\}$, suppose $L_k - L_1 \in d\mathbb{Z}$ for some $d>0$, for $k=1,\dots,N$ (referred to as the \emph{commensurability} condition from hereon). Then for every nonzero integer $m$, 
                \[
                \omega_m=\frac{2\pi m}{d}
                \Longrightarrow
                \mathcal H(\omega_m)=0,
                \]
                independent of $\pi_k$.
        \end{enumerate}
	\end{corollary}
    We defer the proof to \cref{app:cor1}. The proof of \cref{prop:N_path_reflection} in \cref{app:prop_4} introduces a probabilistic framework that makes the converse direction in part 1 of this corollary more accessible. The corollary separates exact broadband transparency (part 1) from a periodic family of isolated transparency frequencies $\omega_m=2\pi m/d$ (part 2). Note that the part-2 mechanism is dictated solely by commensurability of the path-length differences, and the resulting frequencies are \emph{independent} of the width weighting $\pi_k$.

    We stress that for general lengths and widths -- outside the commensurability hypothesis of part 2 -- the algebraic structure of $\mathcal H(\omega)$ depends on the number of paths $N$. To investigate this dependence, we use the fact that 
    \begin{equation}
    \mathcal H(\omega) \;=\; 4\sum_{1\le j<k\le N}\pi_j\pi_k\,\sin^2\!\Bigl(\tfrac{\omega(L_j-L_k)}{2}\Bigr)\,e^{-i\omega(L_j+L_k)},
    \label{eq:H_pair_decomp}
    \end{equation}
    (derived in \cref{app:cor1}) in which each pair $\{j,k\}$ contributes a non-negative term (via the sine squared term) set by the length difference $L_j-L_k$ and the unit-modulus complex exponential by the length sum $L_j+L_k$. Cancellation of $\mathcal H$ at a given $\omega$ thus arises either when every sine squared vanishes simultaneously -- the commensurability mechanism of part 2 -- or when the weighted complex-exponential sum vanishes through phase cancellation. When $N=2$, the double sum collapses into a single term, and thus the magnitude of $\mathcal{H}$ reads
    \[
    |\mathcal H(\omega)| \;=\; 4\pi_1\pi_2\,\sin^2\!\Bigl(\tfrac{\omega(L_2-L_1)}{2}\Bigr) \qquad (N=2),
    \]
    with zeros only at the commensurability frequencies of part 2. At $N=3$, phase cancellation becomes algebraically possible; for $N\geq 4$, an additional hybrid mechanism arises. Both detailed structures for $N=3,4$, including a graphical illustration of these configurations in \cref{fig:N3_reflection_factor}, are collected in \cref{app:frequency_cancellation}.
  	
	\section{Conclusions and extensions}
	\label{sec:conclusion}
	
	We have defined long-wave transparency on channel networks in terms of the vertex scattering matrix.  For a vertex of arbitrary degree, the channel-width-weighted scattering law coefficients take the closed form
	\[
	S_{ab}=\frac{2B_b}{\sum_e B_e}-\delta_{ab}.
	\]
	This formula unifies converging and diverging junctions with the same degree: the scattering matrix is the same, while the assignment of incoming and outgoing characteristic components depends on the local orientation.  Balance is therefore most naturally defined relative to a source set.  For a prescribed source set \(\Sigma\), balance is equivalent to rank-deficiency of the source-source submatrix reflection block \(S_{\Sigma,\Sigma}\), and the transparent source mode is the synchronized direction \(\mathbf 1_\Sigma\).
	
	This local result does not by itself imply transparency of a cyclic network. For the two-path island, local balance at both junctions removes reflection only when the two arc travel times are equal.  More generally, for an \(N\)-path balanced island, the first-generation upstream reflection, observed on the upstream leading-edge, has frequency factor \(P(\omega)^2-P(2\omega)\), where
	\[
	P(\omega)=\sum_{k=1}^N \pi_k e^{-i\omega L_k}, \quad \pi_k = \frac{B_k}{\sum_jB_j}.
	\]
	Exact broadband cancellation forces all path lengths, equivalently all travel times in the present normalization, to agree. Thus local balance is a vertex condition, while transparency also requires synchronization across the paths. This path-delay formulation also suggests a possible probabilistic viewpoint on more general networks (see \cref{app:prop_4} for a discussion), in which width- or admittance-weighted delay measures are propagated through repeated scattering events.

    Although the derivation here uses the linear shallow-water Stoker coupling, the same scattering viewpoint can be applied to other linearized wave networks once edge propagation has been diagonalized into characteristic propagating modes and junctions have been represented by a scattering law. In such settings, the width-weights in the scattering matrix $S$ should be replaced by the appropriate physical admittances, and the synchronization condition should be expressed in terms of the  travel times of the characteristic propagating modes.
    
    Several extensions of the present hyperbolic model remain natural. For nonzero subcritical base flow, $0<U<C_0$, the upstream and downstream characteristic speeds have different magnitudes. The linearized discharge at a junction includes the advective contribution from the base flow so the width-weighted flux condition in \cref{eq:general_flux} must be modified. Furthermore, extending the source-relative balance and island formulas to this case therefore requires rederiving the vertex scattering law with the modified characteristic variables and flux condition. Another natural step is to also allow edge-dependent depths $H_k$ with $U=0$. Then the wave speed on edge $k$ is $c_k=\sqrt{gH_k}$, the travel time is $L_k/c_k$, and the natural balance weights should be hydraulic admittances, such as $B_k c_k$, rather than widths alone. These extensions would test whether the source-relative balance and synchronization principles persist beyond the constant-depth normalization used here.

    \section*{Acknowledgements}
    B. Gu and A. Nachbin would like to thank the Harold J. Gay Professorship in the Department of Mathematical Sciences at Worcester Polytechnic Institute.
        
	\appendix
	
	\section{Converging and Diverging Junction Scattering Equivalence}\label{app:cjdj}
	
	We give the explicit calculation that produces the unified scattering matrix $S(B_i, B_j, B_k)$ in \cref{sec:model} from the converging and diverging configurations of  \cref{2junctions}, showing that the two configurations yield the same matrix despite having different identification of incoming and outgoing modes.
	
	The two configurations (per \cref{2junctions}) differ only in which mode on each edge propagates toward the vertex and which propagates away: at a CJ the outgoing modes are $(v_i,v_j,z_k)$ and the incoming modes are $(z_i,z_j,v_k)$; at a DJ the outgoing modes are $(v_i,z_j,z_k)$ and the incoming modes are $(z_i,v_j,v_k)$. Inserting these labels into the Y-junction conditions \cref{eq:flux_cons_char,eq:cont_char} and rearranging into ``outgoing on the left, incoming on the right'' yields a constraint system
	\begin{equation}
		D_{\rm out}\, \mathbf{x}_{\rm out} = D_{\rm in}\, \mathbf{x}_{\rm in},
		\label{TypeICCs}
	\end{equation}
	in which the coefficient matrices $D_{\rm out}$ and $D_{\rm in}$ are \emph{symbolically} different for CJ and DJ -- their entries involve $m_{11},m_{12},m_{21},m_{22}$ in different positions, because relabeling $v_e\leftrightarrow z_e$ as incoming or outgoing permutes which modal coefficient appears in each row. Substituting the values from \cref{eq:eigenvector} -- namely $m_{11}=-\sqrt{gH}=-m_{12}$ and $m_{21}=H=m_{22}$ -- collapses precisely those entries on which CJ and DJ disagree, so the two systems reduce to the same \emph{numerical} relation. Solving it gives
    \begin{equation}
    \mathbf{x}_{\rm out}(t) \;=\; S(B_i,B_j,B_k)\,\mathbf{x}_{\rm in}(t),
    \label{eq:scattering}
    \end{equation}
    with
    \begin{gather}
    S(B_i,B_j,B_k) \;=\; \frac{1}{B_i+B_j+B_k} \notag \\  
    \begin{bmatrix}
    B_i-B_j-B_k & 2 B_j & 2 B_k \\
    2 B_i & -B_i+B_j-B_k & 2 B_k \\
    2 B_i & 2 B_j & -B_i-B_j+B_k
    \end{bmatrix},
    \label{eq:Sgeneral}
    \end{gather}
    identical for both configurations. The CJ/DJ equivalence is therefore a feature of the specific eigenvector matrix $M$ produced by the shallow-water diagonalization, not of the Y-junction conditions in isolation.        
	
	\section{Proofs}\label{app:proofs}
	\subsection{Proof of \cref{prop:island_sync}}\label[appendix]{app:prop3}
	\begin{proposition*}[Two-path island synchronization (see \cref{2junctions})]
		\label{app_prop:island_sync}
		Suppose vertex-balance holds at both \(\nu_D\) and \(\nu_C\), so that
		\[
		B_0=B_\infty=B_T+B_B.
		\]
		For a nonzero non-periodic incoming pulse on \(e_0\), with amplitude \(a(t)\) at the vertex, the two-path
		island is transparent for all \(t\) if and only if
		\[
		L_T=L_B.
		\]
		If \(L_T\neq L_B\), the first-generation upstream reflection on \(e_0\) is
		\begin{gather*}
			\beta_0^{(1)}(t)
			=
			\frac{B_TB_B}{(B_T+B_B)^2}
			\times \notag \\ \left[
			2a(t-L_T-L_B)-a(t-2L_T)-a(t-2L_B)
			\right].
		\end{gather*}
	\end{proposition*}
	\begin{proof}
		At \(\nu_D\), the active source set is \(\Sigma_D=\{e_0\}\).  The balance condition \(B_0=B_T+B_B\) gives no reflection on \(e_0\).  The full scattering law gives the outgoing characteristic amplitudes on the two arc edges,
		\[
		\beta_T(t)=\beta_B(t)=a(t).
		\]
		After propagation along the two arcs, the incoming source vector at
		\(\nu_C\) is
		\[
		\alpha_{\Sigma_C}(t)
		=
		\begin{pmatrix}
			a(t-L_T)\\
			a(t-L_B)
		\end{pmatrix},
		\qquad
		\Sigma_C=\{e_T,e_B\}.
		\]
		At \(\nu_C\), balance means \(B_\infty=B_T+B_B\).  By
		\cref{prop:balance_rank}, the source-source reflection block
		has kernel direction
		\[
		\mathbf 1_{\Sigma_C}=(1,1)^T.
		\]
		Therefore $\nu_C$ generates no back-reflection precisely when
		\[
		\alpha_{\Sigma_C}(t)\in \operatorname{span}\{\mathbf 1_{\Sigma_C}\}
		\qquad \text{for all }t,
		\]
		or equivalently,
		\[
		a(t-L_T)=a(t-L_B), \quad \forall t>0 \iff L_T = L_B.
		\]
		If \(L_T\neq L_B\), the two arc signals reflected from \(\nu_C\) propagate back to \(\nu_D\).  
		
		At the downstream vertex \(\nu_C\), use the edge order $(e_T,e_B,e_\infty)$. By \cref{prop:vertex_scattering}, the scattering law at \(\nu_C\) is
		\[
		\begin{pmatrix}
			\beta_T^C\\
			\beta_B^C\\
			\beta_\infty^C
		\end{pmatrix}
		=
		\begin{pmatrix}
			-\dfrac{B_B}{B_T+B_B} & \dfrac{B_B}{B_T+B_B} & 1\\[0.7em]
			\dfrac{B_T}{B_T+B_B} & -\dfrac{B_T}{B_T+B_B} & 1\\[0.7em]
			\dfrac{B_T}{B_T+B_B} & \dfrac{B_B}{B_T+B_B} & 0
		\end{pmatrix}
		\begin{pmatrix}
			\alpha_T^C\\
			\alpha_B^C\\
			\alpha_\infty^C
		\end{pmatrix}.
		\]
		During the first recombination event,
		\[
		\alpha_T^C(t)=a(t-L_T),
		\quad
		\alpha_B^C(t)=a(t-L_B),
		\quad
		\alpha_\infty^C(t)=0.
		\]
		Therefore the reflected amplitudes sent back into the two arcs are
		\[
		\beta_T^C(t)
		=
		\frac{B_B}{B_T+B_B}
		\left[
		a(t-L_B)-a(t-L_T)
		\right],
		\]
		and
		\[
		\beta_B^C(t)
		=
		\frac{B_T}{B_T+B_B}
		\left[
		a(t-L_T)-a(t-L_B)
		\right].
		\]
		
		At the upstream vertex \(\nu_D\), we use the ordering of
		edge types $(e_T,e_B,e_0)$.  By \cref{prop:vertex_scattering}, we have
		\[
		\begin{pmatrix}
			\beta_T^D\\
			\beta_B^D\\
			\beta_0^D
		\end{pmatrix}
		=
		\begin{pmatrix}
			-\dfrac{B_B}{B_T+B_B} & \dfrac{B_B}{B_T+B_B} & 1\\[0.7em]
			\dfrac{B_T}{B_T+B_B} & -\dfrac{B_T}{B_T+B_B} & 1\\[0.7em]
			\dfrac{B_T}{B_T+B_B} & \dfrac{B_B}{B_T+B_B} & 0
		\end{pmatrix}
		\begin{pmatrix}
			\alpha_T^D\\
			\alpha_B^D\\
			\alpha_0^D
		\end{pmatrix}.
		\]
		For the first-generation reflected group returning from \(\nu_C\), the upstream edge is quiescent, so
		\[
		\alpha_0^D(t)=0.
		\]
		The two incoming arc amplitudes at \(\nu_D\) are the reflected amplitudes from \(\nu_C\), delayed by their respective return times:
		\[
		\alpha_T^D(t)=\beta_T^C(t-L_T),
		\qquad
		\alpha_B^D(t)=\beta_B^C(t-L_B).
		\]
		Therefore the upstream reflected amplitude is the third component of the same scattering law:
		\begin{gather*}
			\beta_0^{(1)}(t)
			=
			\beta_0^D(t)
			\\ = 
			\frac{B_T}{B_T+B_B}\,\beta_T^C(t-L_T)
			+
			\frac{B_B}{B_T+B_B}\,\beta_B^C(t-L_B).
		\end{gather*}
		Substituting the expressions for \(\beta_T^C\) and \(\beta_B^C\) gives
		\[
		\begin{aligned}
			\beta_0^{(1)}(t)
			&=
			\frac{B_TB_B}{(B_T+B_B)^2}
			\left[
			a(t-L_T-L_B)-a(t-2L_T)
			\right] \\
			&\quad+
			\frac{B_TB_B}{(B_T+B_B)^2}
			\left[
			a(t-L_T-L_B)-a(t-2L_B)
			\right],
		\end{aligned}
		\]
		hence
		\begin{gather*}
			\beta_0^{(1)}(t)
			=
			\frac{B_TB_B}{(B_T+B_B)^2} \times \\
			\left[
			2a(t-L_T-L_B)-a(t-2L_T)-a(t-2L_B)
			\right].
		\end{gather*}
	\end{proof}
	
	\subsection{Proof of \cref{prop:N_path_reflection}}\label[appendix]{app:prop_4}
	\begin{proposition*}[First-generation upstream reflection for an \(N\)-path island]
		Under the balance condition \cref{eq:N_balance}, the first-generation upstream reflection on \(e_0\) is
		\begin{equation}
        \begin{aligned}
        \beta_0^{(1)}(t)
        &=
        \frac{1}{A^2}
        \sum_{k=1}^{N-1}\sum_{j=k+1}^{N}
        \Biggl\{
        B_kB_j \times {}\\
        &\qquad\left[
        2a(t-L_k-L_j)-a(t-2L_k)-a(t-2L_j)
        \right]
        \Biggr\}.
        \end{aligned}
        \end{equation}
        where $A = \sum^N_{k=1} B_k$ is the total width of all internal edges.
		Equivalently, if
		\[
		\pi_k=\frac{B_k}{A},
		\qquad
		P(\omega)=\sum_{k=1}^N \pi_k e^{-i\omega L_k},
		\]
		then in frequency domain,
		\begin{equation}
			\widehat{\beta}_0^{(1)}(\omega)
			=
			\widehat a(\omega)
			\left[
			P(\omega)^2-P(2\omega)
			\right].
			\label{eq:app_beta0_N_freq}
		\end{equation}
	\end{proposition*}
	\begin{proof}
		At \(\nu_D\), balance gives no reflection on \(e_0\) and sends the same
		outgoing characteristic amplitude into each arc:
		\[
		\beta_k^D(t)=a(t),
		\qquad
		k=1,\dots,N.
		\]
		After propagation, the incoming data at \(\nu_C\) are
		\[
		\alpha_k^C(t)=a(t-L_k),
		\qquad
		\alpha_\infty^C(t)=0.
		\]
		Using the weighted-mean form \cref{eq:weighted_mean_scattering}, the mean at
		\(\nu_C\) is
		\[
		\overline{\alpha}^C(t)
		=
		\frac{1}{2A}
		\sum_{j=1}^N B_j a(t-L_j).
		\]
		Therefore the back-reflection sent from \(\nu_C\) into arc \(e_k\) is
		\begin{align}
		\beta_k^C(t)
		&=
		2\overline{\alpha}^C(t)-\alpha_k^C(t) \notag \\
		&=
		\left[\frac{1}{A}\sum_{j=1}^N B_j a(t-L_j)\right]-a(t-L_k). \label{eq:beta_C_k}
		\end{align}
		This reflected signal returns to \(\nu_D\) after another delay \(L_k\), so
		\[
		\alpha_k^{D,2}(t)=\beta_k^C(t-L_k).
		\]
		With \(e_0\) quiescent during this return, the outgoing reflected component
		on \(e_0\) is
		\begin{align}
		\beta_0^{(1)}(t) &= 2\overline{\alpha}^D(t) - \cancelto{0}{\alpha_0^D(t)} \notag\\
        &= 2 \left[\frac{1}{2A} \sum_{k=1}^N B_k\beta_k^C(t-L_k)\right] \notag\\
		&= \frac{1}{A} \sum_{k=1}^N B_k\beta_k^C(t-L_k). \label{eq:beta_1_0_total}
		\end{align}
		Substituting \cref{eq:beta_C_k} into \cref{eq:beta_1_0_total} gives
		\[
        \begin{aligned}
		\beta_0^{(1)}(t) 
		&=
		\frac{1}{A^{2}}\sum_{k=1}^{N}B_{k}\sum_{j=1}^{N}B_{j}a\left(t-L_{k}-L_{j}\right) \\ 
		&-
		\frac{1}{A}
		\sum_{k=1}^N
		B_k a(t-2L_k),
        \end{aligned}
		\]		
        with now the goal of combining the two sums. 
        
		We split the double sum into off-diagonal and diagonal parts:
		\begin{align}
			\beta_0^{(1)}(t)
			&=
			\frac{1}{A^2}
			\sum_{k=1}^{N}B_{k}\sum_{j\neq k}^{N}B_{j}a\left(t-L_{k}-L_{j}\right)
			\notag \\
			&\quad+
			\frac{1}{A^2}
			\sum_{k=1}^N
			B_k^2 a(t-2L_k)
			-
			\frac{1}{A}
			\sum_{k=1}^N
			B_k a(t-2L_k). \label{eq:separate_beta}
		\end{align}
        Looking at the diagonal terms in the second line, we use the fact that \(A=\sum_{j=1}^N B_j\) and deduce that the coefficient becomes
		\begin{gather}
        \frac{B_{k}^{2}}{A^{2}}-\frac{B_{k}}{A}=\frac{B_{k}}{A^{2}}\left(B_{k}-A\right)\notag \\
        =\frac{B_{k}}{A^{2}}\left(B_{k}-\sum_{j=1}^{N}B_{j}\right)=-\frac{B_{k}}{A^{2}}\sum_{j\neq k}^{N}B_{j}. \label{eq:derived_coefficient}
        \end{gather}
		Substituting \cref{eq:derived_coefficient} into \cref{eq:separate_beta} yields
		\begin{align*}
        \beta_{0}^{(1)}(t) & =\frac{1}{A^{2}}\sum_{k=1}^{N}B_{k}\sum_{j\neq k}^{N}B_{j}\left[a\left(t-L_{k}-L_{j}\right)-a\left(t-2L_{k}\right)\right]\\
         &=
        \frac{1}{A^2}
        \sum_{k=1}^{N-1}\sum_{j=k+1}^{N}
        \Biggl\{
        B_kB_j \times {}\\
        &\qquad\left[
        2a(t-L_k-L_j)-a(t-2L_k)-a(t-2L_j)
        \right]
        \Biggr\}.
        \end{align*}
        Here the last equality is obtained by splitting the ordered sum into the regions \(j<k\) and \(j>k\), then relabeling the \(j<k\) contribution so that both parts are indexed by \(1\le k<j\le N\).
		
		In frequency domain, we take the Fourier transform of \cref{eq:beta0_N_time} in time, with time shifts mapped to multiplication by \(e^{-i\omega L}\), which gives
		\begin{align}
		\widehat{\beta}_0^{(1)}(\omega)
		&=
		-\widehat a(\omega)
		\sum_{k=1}^{N-1}\sum_{j=k+1}^{N}
		\pi_k\pi_j
		\left(
		e^{-i\omega L_k}-e^{-i\omega L_j}
		\right)^2 \notag \\ 
        &=
		-\frac{\widehat a(\omega)}{2}
		\sum_{k=1}^{N}\sum_{j=1}^{N}
		\pi_k\pi_j
		\left(
		e^{-i\omega L_k}-e^{-i\omega L_j}
		\right)^2 .
        \label{eq:beta_hat_pix}
		\end{align}
        We collapse this double sum with a probabilistic re-expression; this is a calculation tool, and the interpretation it makes possible is recorded as a 
        brief remark after the proof. 
        
        Let \(X\) be a random variable distributed according to the width-weighted path-delay measure
        \[
        \mu \;=\; \sum_{k=1}^{N} \pi_k\,\delta_{L_k},
        \qquad
        \pi_k=\frac{B_k}{A},
        \]
        and let \(Y\) be an independent copy of \(X\), i.e. $X\overset{d}{\sim}Y$. Writing \(f(x):=e^{-i\omega x}\), the double sum in \cref{eq:beta_hat_pix} becomes
        \begin{gather*}
        \sum_{k=1}^{N}\sum_{j=1}^{N}
		\pi_k\pi_j
		\left(f(L_k)-f(L_j)
		\right)^2
        = \mathbb{E}_{\mu}\!\left[\bigl(f(X)-f(Y)\bigr)^2\right] \\
        = 2\!\left(\mathbb{E}_{\mu}[f(X)^2] - \mathbb{E}_{\mu}[f(X)]\,\mathbb{E}_{\mu}[f(Y)]\right) \\
        = 2\!\left(P(2\omega)-P(\omega)^2\right),
        \end{gather*}
        where the last line uses \(\mathbb{E}[f(X)]=P(\omega)\), \(\mathbb{E}[f(X)^2]=\mathbb{E}[e^{-i 2\omega X}]=P(2\omega)\), and the independence of \(X\) and \(Y\) to factor \(\mathbb{E}[f(X)\,f(Y)]=P(\omega)^2\).  Substituting into \cref{eq:beta_hat_pix} yields \cref{eq:app_beta0_N_freq}.
	\end{proof}

    \paragraph*{Probabilistic interpretation of \(\mathcal{H}(\omega)\).}
    The substitution above identifies \(P(\omega)=\mathbb{E}_{\mu}[e^{-i\omega X}]\) as the characteristic function of the path-delay measure \(\mu\).  In these terms \(P(\omega)^2\) is the characteristic function of \(X+Y\) -- two independently selected path delays -- and \(P(2\omega)\) is the characteristic function of \(2X\) -- twice the same path delay.  These are exactly the mixed-path and same-path return-delay contributions in \cref{eq:beta0_N_time}, so \(\mathcal{H}(\omega)\) measures, in the Fourier domain, the mismatch between mixed-path and same-path first-generation returns.

    \subsection{Proof of \cref{cor:codim_lengths}}\label[appendix]{app:cor1}
    \begin{corollary*}[Cancellation mechanisms for an $N$-path Island]
        Let 
        \[
        \mathcal{H}(\omega):=P(\omega)^2-P(2\omega) ,   
        \]
        where $P(\omega)$ is defined in \cref{eq:frequency_P}.
		Then the following statements hold.
        \begin{enumerate}
            \item $\mathcal H(\omega)=0 \ \text{for all } \omega\in\mathbb R$ $\Longleftrightarrow$ $L_1=L_2=\cdots=L_N$.
            \item For unequal lengths $\{L_k\}$, suppose $L_k - L_1 \in d\mathbb{Z}$ for some $d>0$, for $k=1,\dots,N$. Then for every nonzero integer $m$, 
                \[
                \omega_m=\frac{2\pi m}{d}
                \Longrightarrow
                \mathcal H(\omega_m)=0,
                \]
                independent of $\pi_k$.
        \end{enumerate}
	\end{corollary*}
    \begin{proof}       
        We retain the notation of the proof of \cref{prop:N_path_reflection}: \(X\) is the path-delay random variable with law \(\mu=\sum_{k=1}^N \pi_k\,\delta_{L_k}\), and \(Y\) is an independent copy of \(X\). 

        \textbf{Part 1}:
        
        (Direction $\Leftarrow$) If \(L_1=\cdots=L_N=L\), then \(X=L\), \(P(\omega)=e^{-i\omega L}\), and \(P(\omega)^2=e^{-2i\omega L}=P(2\omega)\) for every \(\omega\in\mathbb R\); hence \(\mathcal H\equiv 0\).
        
        Conversely (direction $\Rightarrow$), suppose \(\mathcal H(\omega)=0\) for every \(\omega\in\mathbb R\). Then \(P(\omega)^2=P(2\omega)\) on all of \(\mathbb R\), so the characteristic functions of \(X+Y\) and \(2X\) coincide everywhere. By the uniqueness theorem for characteristic functions, \(X+Y\stackrel{d}{=}2X\). Equating variances and using independence of \(X\) and \(Y\),
        \[
        2\,\operatorname{Var}(X)
        \;=\;
        \operatorname{Var}(X+Y)
        \;=\;
        \operatorname{Var}(2X)
        \;=\;
        4\,\operatorname{Var}(X),
        \]
        so \(\operatorname{Var}(X)=0\), and \(X\) is a constant. Hence \(\mu\) is a single point mass, and \(L_1=L_2=\cdots=L_N\).

        \textbf{Part 2}:
        
        Applying the half-angle identity
        \[
        e^{-i\omega L_j} - e^{-i\omega L_k} \;=\; -2i\,\sin\!\Bigl(\tfrac{\omega(L_j-L_k)}{2}\Bigr)\,e^{-i\omega(L_j+L_k)/2}
        \]
        to the pairwise expansion, we have
        \begin{align}
        \mathcal H(\omega) &= -\sum_{j<k}\pi_j\pi_k\bigl(e^{-i\omega L_j} - e^{-i\omega L_k}\bigr)^2 \notag \\
        &= 4\sum_{1\le j<k\le N}\pi_j\pi_k\,\sin^2\!\Bigl(\tfrac{\omega(L_j-L_k)}{2}\Bigr)\,e^{-i\omega(L_j+L_k)}. \label{eq:sine_squared}
        \end{align}
        By the hypothesis $L_k - L_1 \in d\mathbb Z$ for every $k$, every pairwise difference satisfies $L_j - L_k = (L_j-L_1) - (L_k-L_1) \in d\mathbb Z$; write $L_j - L_k = n_{jk}\,d$ with $n_{jk}\in\mathbb Z$. At $\omega_m = 2\pi m/d$,
        \[
        \frac{\omega_m(L_j-L_k)}{2} \;=\; \pi\,m\,n_{jk} \;\in\; \pi\mathbb Z,
        \]
        so $\sin\!\bigl(\omega_m(L_j-L_k)/2\bigr) = 0$ for every pair $\{j,k\}$. Every term in \cref{eq:sine_squared} vanishes, hence $\mathcal H(\omega_m) = 0$, independent of $\{\pi_k\}$.
    \end{proof}
    
    We comment on the behavior of $\mathcal{H}(\omega)$ near $\omega=0$ by Taylor expanding
    \[
    P(\omega) \;=\; \mathbb{E}[e^{-i\omega X}] \;=\; 1 \;-\; i\omega\,\mathbb{E}[X] \;-\; \tfrac{\omega^2}{2}\,\mathbb{E}[X^2] \;+\; O(\omega^3),
    \]
    squaring through order \(\omega^2\) gives
    \[
    P(\omega)^2
    \;=\;
    1 
    \;-\; 2i\omega\,\mathbb{E}[X]
    \;-\; \omega^2\,\bigl(\mathbb{E}[X^2] + \mathbb{E}[X]^2\bigr)
    \;+\; O(\omega^3),
    \]
    while substituting \(\omega\mapsto 2\omega\) in the same expansion gives
    \[
    P(2\omega) \;=\; 1 \;-\; 2i\omega\,\mathbb{E}[X] \;-\; 2\omega^2\,\mathbb{E}[X^2] \;+\; O(\omega^3).
    \]
    The constant and linear-in-\(\omega\) terms encode only the first moment \(\mathbb{E}[X]\) and agree between the two expressions; they cancel in the difference. What survives is the mismatch between the second-moment coefficients: \(2\mathbb{E}[X^2] - \bigl(\mathbb{E}[X^2]+\mathbb{E}[X]^2\bigr) = \mathbb{E}[X^2]-\mathbb{E}[X]^2\). Hence
    \begin{align}
    \mathcal{H}(\omega) 
    \;=\; 
    P(\omega)^2 - P(2\omega)
    &\;=\; 
    \omega^2\,\bigl(\mathbb{E}[X^2]-\mathbb{E}[X]^2\bigr) + O(\omega^3) \notag \\
    &\;=\;
    \omega^2\,\operatorname{Var}(X) + O(\omega^3), \label{eq:variance}
    \end{align}
    identifying the variance of the path-delay random variable -- equivalently, the \(\pi\)-weighted variance \(\sum_k\pi_k L_k^2-(\sum_k\pi_k L_k)^2\) of the path-length tuple -- as the leading low-frequency obstruction to first-generation transparency. 
	
	\section{Numerical Method and Results}\label{sec:numerics}
	
	\subsection{Lagrangian scheme}
	This appendix records the characteristic propagation scheme used to generate numerical visualizations. The analytical results in the main text rely on the exact characteristic scattering laws; the numerical scheme simply implements those laws on a grid.
	
	The simulations use the characteristic variables introduced in \cref{uni-equ}.  Since the linearized system diagonalizes into two constant-speed transport equations, the numerical update is performed mode-by-mode along characteristics and the junction update is imposed through the scattering matrix per \cref{eq:Sgeneraldegree}.  This avoids solving a global linear system on the graph at each time step; implicit finite-difference and Crank--Nicolson approaches for related network problems may be found in, for example, Jacovkis~\cite{Jacovkis90} and the quantum-graph computations of Goodman et al.~\cite{QGLAB} and Yusupov et al.~\cite{Yusupov1,Yusupov2,Yusupov2025}.
	
	Along the characteristics,
	\begin{align}
		\frac{dv}{dt} (x_1(t),t) &= 0,~ \mbox{with} ~ \frac{dx_1}{dt} = \lambda_1, \\
		\frac{dz}{dt} (x_2(t),t) &= 0, ~ \mbox{with} ~ \frac{dx_2}{dt} = \lambda_2.
		\label{Lagrange}
	\end{align}
	For the simulations shown here we take \(U=0\), so
	\[
	\lambda_1=-C_0,\qquad \lambda_2=C_0,
	\]
	and choose the mesh ratio \(\Delta x/\Delta t=C_0\).  In this case the
	characteristic update is exact on the grid:
	\begin{equation}
		v_j^{n+1}=v_{j+1}^n,\qquad z_j^{n+1}=z_{j-1}^n,
		\label{eq:lagrangian}
	\end{equation}
	up to the junction scattering operations (see \cref{Grid}).  For \(U>0\) in the subcritical regime, \(|\lambda_1|\neq \lambda_2\) in general; an interpolation-based semi-Lagrangian update would then be the natural extension, but is not used in the simulations reported here.

    The numerical figures in \cref{sec:numerics} were generated with custom MATLAB codes implementing the characteristic update and vertex scattering rules described above; the codes are available from the corresponding author upon reasonable request.
	\begin{figure}[h]
		\centering
		\includegraphics[height=1.5in,width=2.2in]{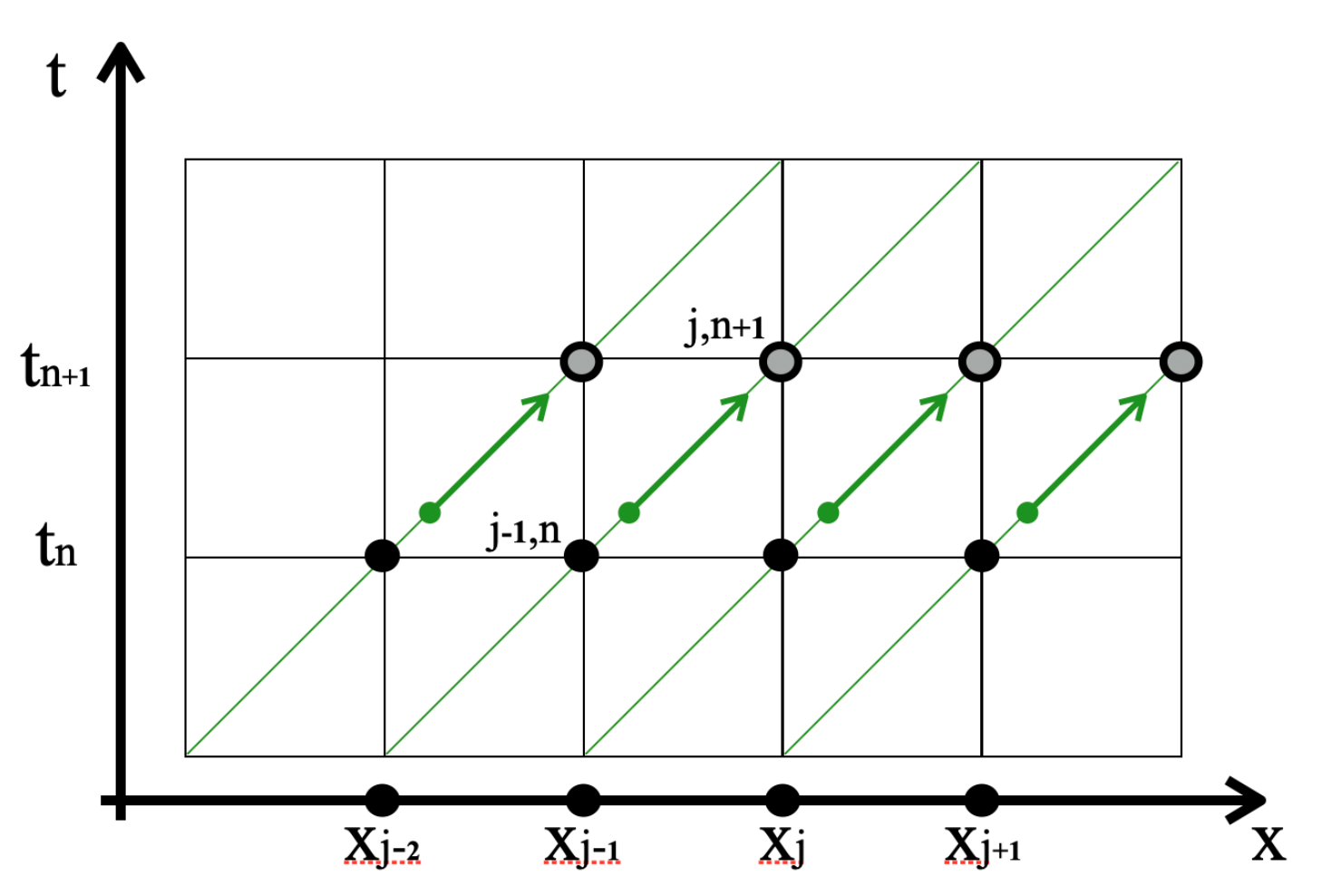}
		\caption{The Lagrangian numerical scheme for the downstream mode: $z_{j}^{n+1}:=z_{j-1}^{n}$.}
		\label{Grid}
	\end{figure}
	
	\subsection{Two-path synchronization}
	We now visualize the synchronization mechanism in \cref{prop:island_sync}. The simulations are not intended as a separate numerical study; they show how the same locally balanced junctions produce either transparency or reflection depending only on the relative arrival times at \(\nu_C\). Consistent with the analytical assumptions in the main text, we take \(U=0\), so the grid update in \cref{eq:lagrangian} is exact between junction scattering events.
	
	In all cases simulated here, we consider an upstream Gaussian disturbance 
	\begin{equation}
		g(x)=\exp\left(-\frac{\left(x-0.2\right)^{2}}{0.005}\right)
		\label{IC}
	\end{equation}
	incoming from the left edge $e_0$ (the Gaussian center and width parameter are chosen so that the pulse fits $e_0$). It interacts first with the junction $\nu_D$, and after transmitting into the arcs (without reflection), it interacts with $\nu_C$. More explicitly, using the scattering matrix at $\nu_D$
	\begin{equation}
		\label{IO_BCJb}
		\begin{bmatrix}
			v_i(t)\\
			z_j(t) \\
			z_k(t)\\
		\end{bmatrix}
		=
		\begin{bmatrix}
			0 &\frac{1}{2} & \frac{1}{2}   \\
			1& -\frac{1}{2} & \frac{1}{2} \\
			1 & \frac{1}{2} & -\frac{1}{2}  \\
		\end{bmatrix}\\ 
		\begin{bmatrix}
			a_0(t)\\
			0 \\
			0\\
		\end{bmatrix}
		= 
		\begin{bmatrix}
			0\\
			a_0(t)\\
			a_0(t)\\
		\end{bmatrix}.
	\end{equation}
	Thus the pulse splits into two identical disturbances, one on each reach of the island. No reflection is observed. This first stage is identical for the symmetric and asymmetric islands; see \cref{IS_t0p8} and \cref{IAS_t0p8}.
	
	\begin{figure*}[!t]
		\centering		
		\begin{minipage}[t]{0.48\textwidth}
			\centering
			\textbf{Symmetric island}\\[0.3em]
			
			\subfloat[$t=0$]{\label{IS_t0}
				\includegraphics[height=1.7in,keepaspectratio]{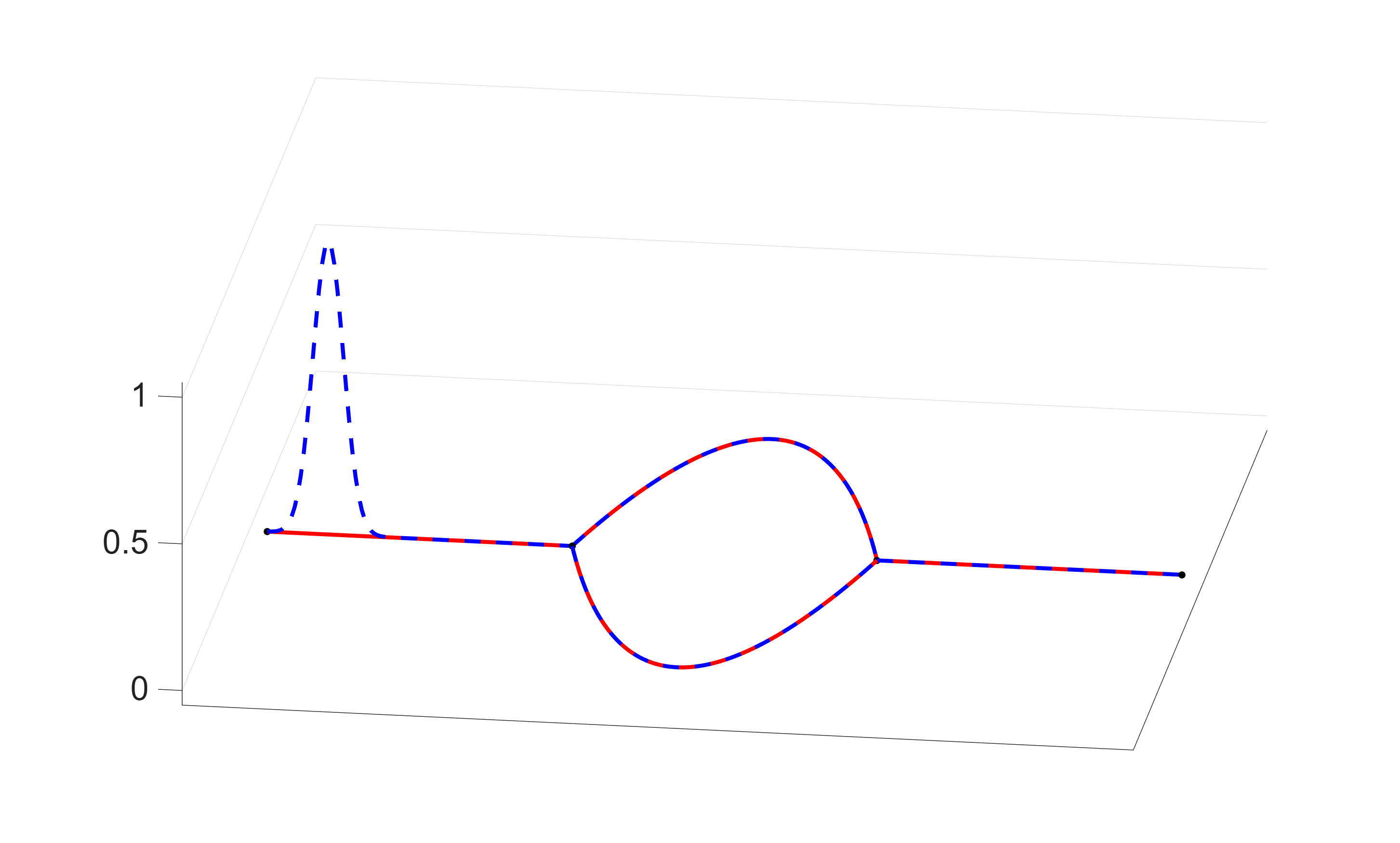}}\\
			
			\subfloat[$t=0.8$]{\label{IS_t0p8}
				\includegraphics[height=1.7in,keepaspectratio]{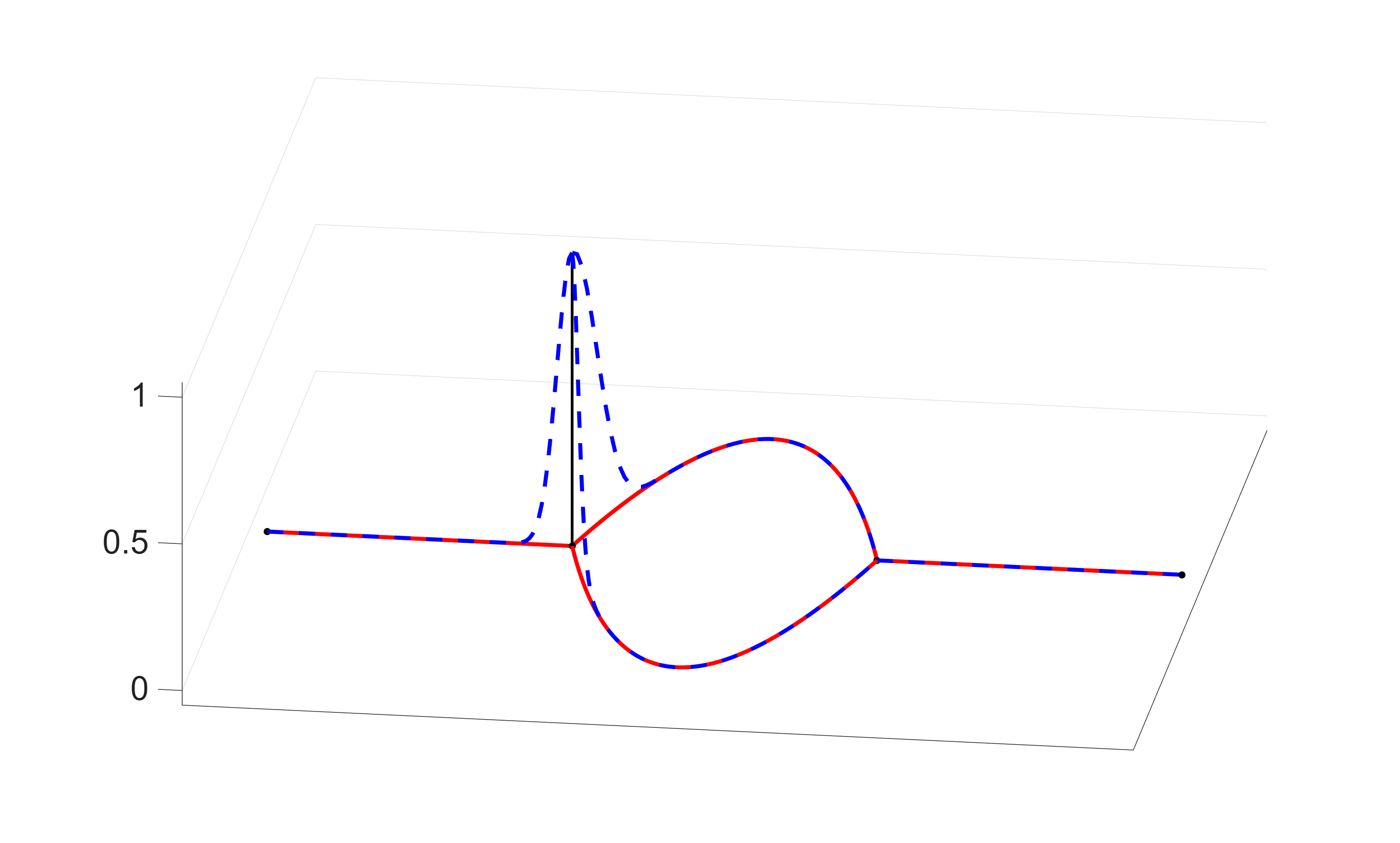}}\\
			
			\subfloat[$t=1.8$]{\label{IS_t1p8}
				\includegraphics[height=1.7in,keepaspectratio]{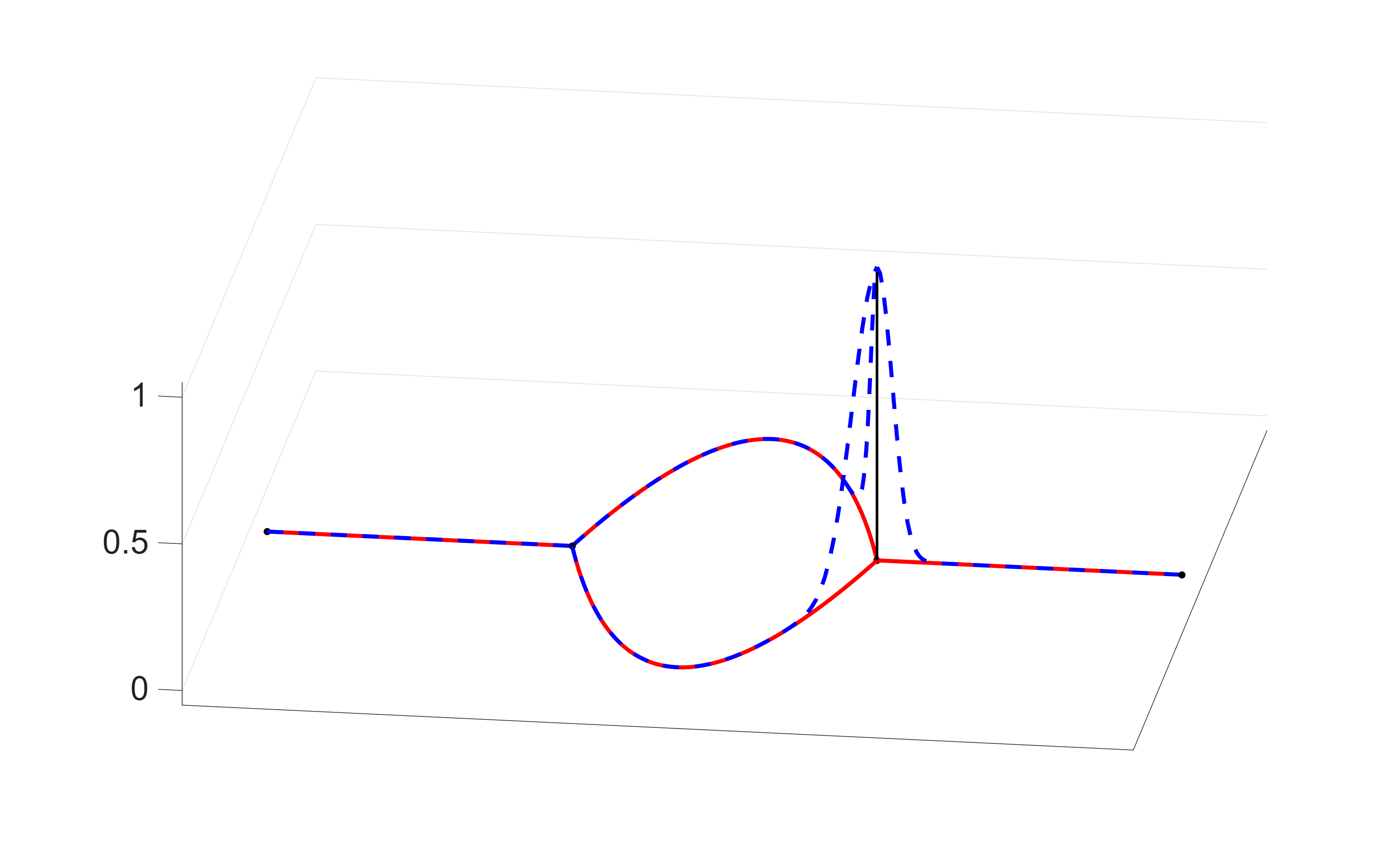}}\\
			
			\subfloat[$t=2.1$]{\label{IS_t2p1}
				\includegraphics[height=1.7in,keepaspectratio]{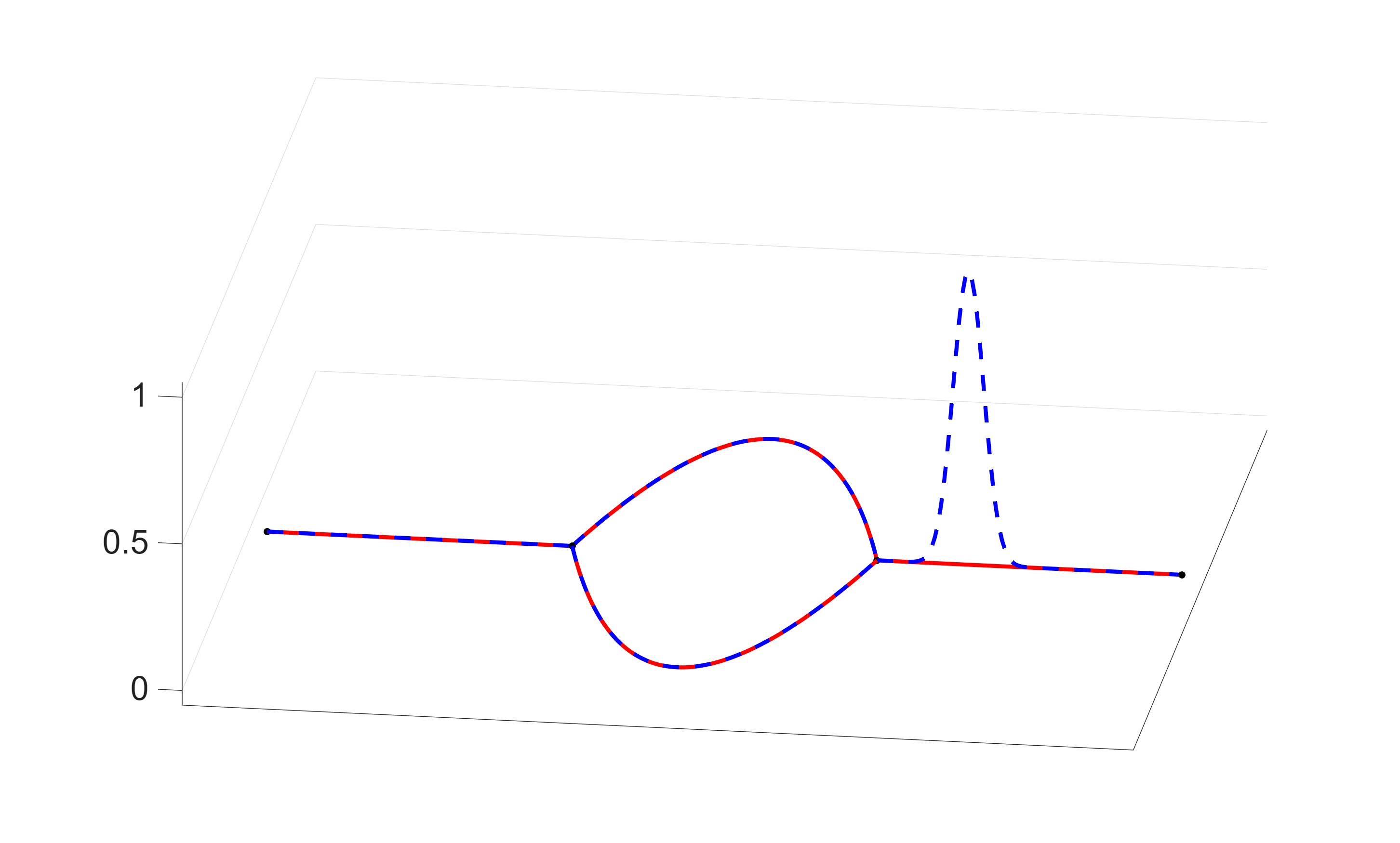}}
		\end{minipage}
		\hfill
		\begin{minipage}[t]{0.48\textwidth}
			\centering
			\textbf{Asymmetric island}\\[0.3em]
			
			\subfloat[$t=0$]{\label{IAS_t0}
				\includegraphics[height=1.7in,keepaspectratio]{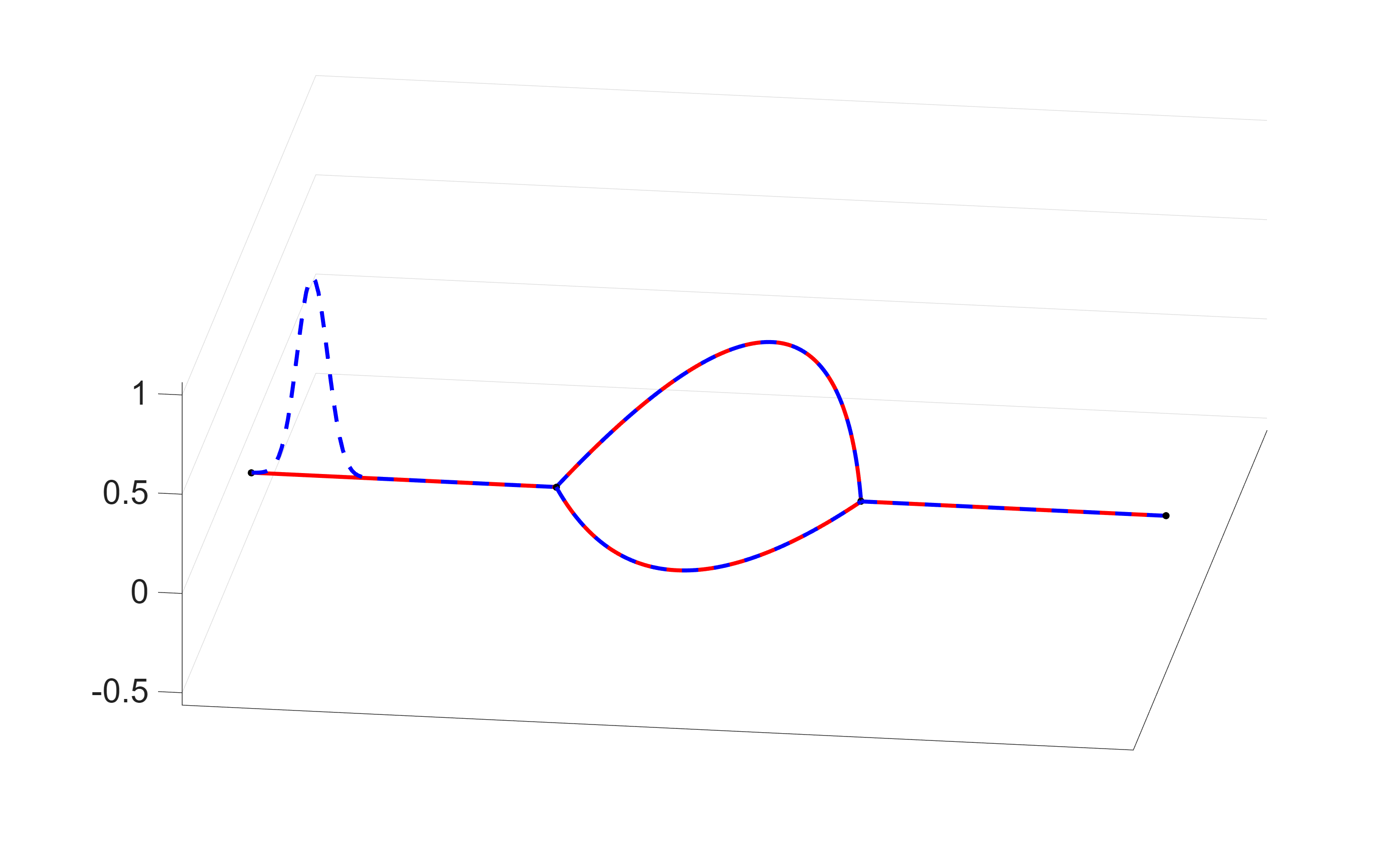}}\\
			
			\subfloat[$t=0.8$]{\label{IAS_t0p8}
				\includegraphics[height=1.7in,keepaspectratio]{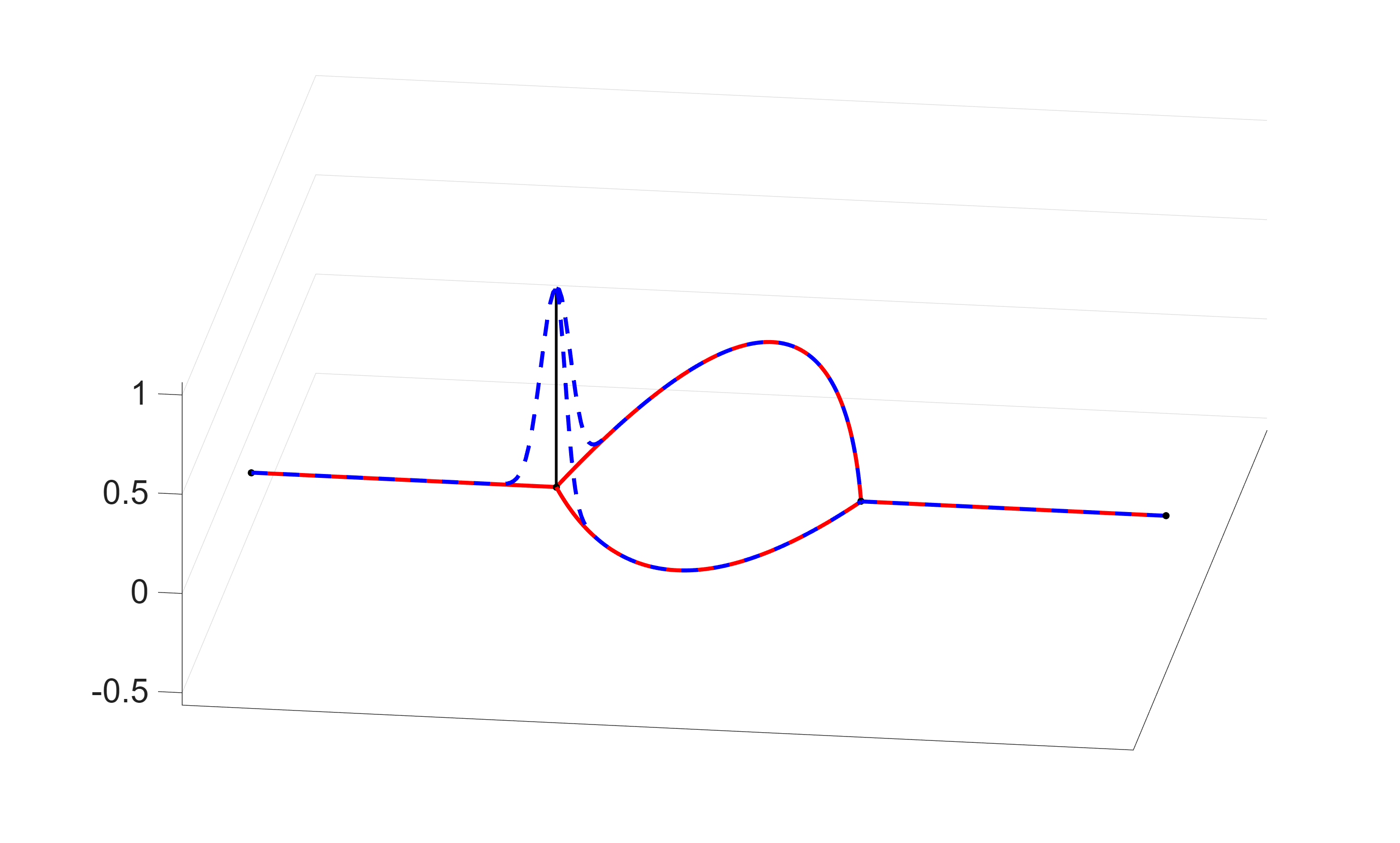}}\\
			
			\subfloat[$t=1.6$]{\label{IAS_t1p6}
				\includegraphics[height=1.7in,keepaspectratio]{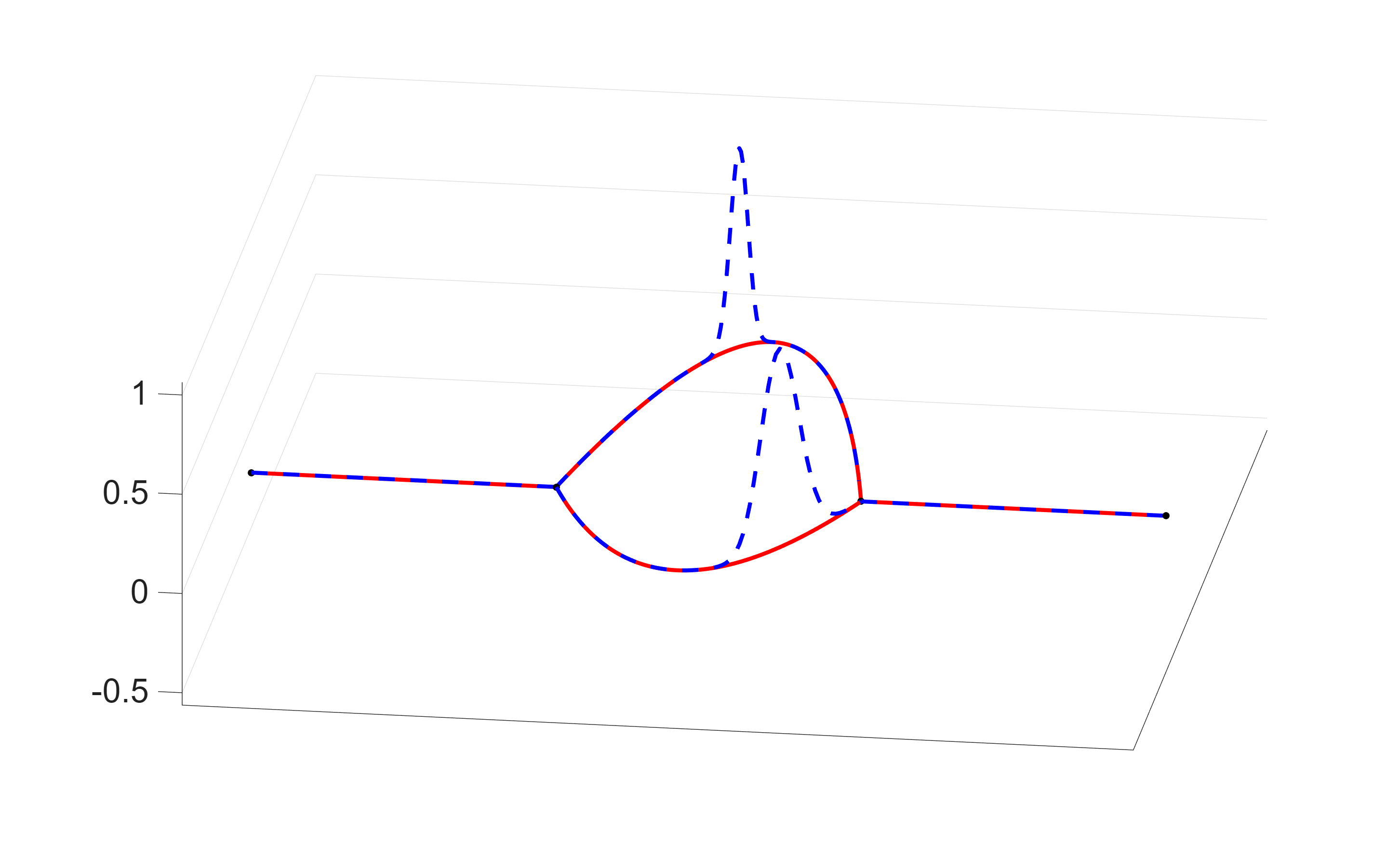}}\\
			
			\subfloat[$t=2.1$]{\label{IAS_t2p1}
				\includegraphics[height=1.7in,keepaspectratio]{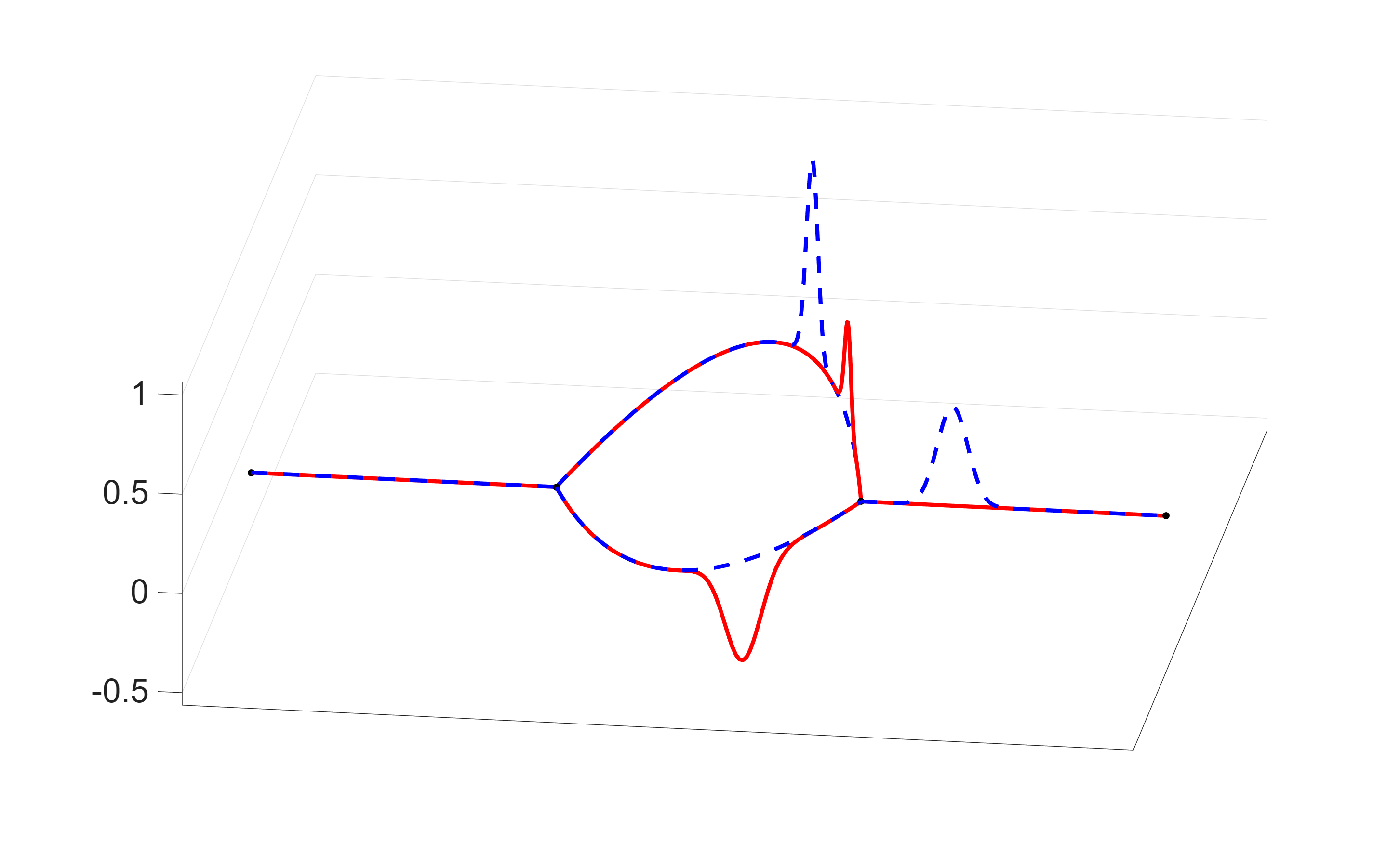}}
		\end{minipage}
		
		\caption{
			Numerical visualization of the two-path island mechanism.  \textbf{Left column}: symmetric island with equal arc lengths. The incoming Gaussian splits at the diverging junction into two identical downstream disturbances, which arrive synchronously at the balanced converging junction and transmit without reflection. \textbf{Right column}: asymmetric island with unequal arc lengths. The initial splitting remains reflectionless, but the longer top edge delays one partial wave; the two arrivals at the converging junction are therefore out of sync, generating upstream reflected waves. Thus local balance at both junctions is sufficient for transparency only when the partial waves arrive synchronously.
		}
		\label{fig:island_sims}
	\end{figure*}
	
	After propagating around the island the respective waves then reach $\nu_C$. Now with the possibility of two different disturbances along the top (T) and bottom (B) reaches of the island, we find
	\begin{equation}
		\label{IO_BCJ_CJ}
		\begin{bmatrix}
			v_i(t)\\
			v_j(t) \\
			z_k(t)\\
		\end{bmatrix}
		=
		\begin{bmatrix}
			-\frac{1}{2} &\frac{1}{2} & 1   \\
			\frac{1}{2}& -\frac{1}{2} & 1 \\
			\frac{1}{2} & \frac{1}{2} & 0  \\
		\end{bmatrix}\\ 
		\begin{bmatrix}
			a_T(t)\\
			a_B(t) \\
			0\\
		\end{bmatrix}
		= 
		\begin{bmatrix}
			\frac{1}{2}(a_B(t)-a_T(t))\\
			\frac{1}{2}(a_T(t)-a_B(t)) \\
			\frac{1}{2}(a_T(t)+a_B(t))\\
		\end{bmatrix}.
	\end{equation}
	This junction is balanced, but it is transparent only when the two incoming traces agree at the same time, \(a_T(t)=a_B(t)\) (see, in particular, \cref{IS_t1p8}).  If the arrivals are not synchronized, \(a_T(t)\neq a_B(t)\), the balanced vertex reflects the nonsynchronized component. In the asymmetric island, the shorter bottom path reaches \(\nu_C\) first, so \cref{IO_BCJ_CJ} predicts reflected components on both arcs; these are visible in \cref{fig:island_sims}. At each arrival, a depression wave is reflected upstream along the incoming reach (reach $j$, or the bottom shorter edge) while an elevation wave is reflected upstream along the quiescent incoming reach (reach $i$, or the upper longer edge). This is clearly observed in \cref{IAS_t2p1}, where at the respective time $a_T(t)=0$. 

    The two-path simulations in \cref{fig:island_sims} visualize the synchronization mechanism in \cref{prop:island_sync}: local balance at the two vertices removes reflection only when the two arc arrivals at \(\nu_C\) are synchronized.  

    \subsection{Degree-four star graph}
    We next simulate a degree-four balanced star graph with two source edges, as a direct check of the source-relative balance condition in \cref{def:balanced_vertex}. The source edges have width $B_1=1/3$ and $B_2=2/3$, while the initially quiescent edges have width $B_3=1/4$ and $B_4=3/4$, respectively. Thus $B_1+B_2=B_3+B_4=1$ so the vertex is balanced relative to the source set $\Sigma = \left\{e_1,e_2\right\}$. Both source edges are initialized with the same Gaussian pulse per \cref{IC}.

    For completeness, by \cref{prop:vertex_scattering}, the scattering matrix of a degree-four star graph, with the prescribed widths, is given by
    \[
    S
    =
    \begin{pmatrix}
    -\frac{2}{3} & \frac{2}{3} & \frac{1}{4} & \frac{3}{4} \\
    \frac{1}{3} & -\frac{1}{3} & \frac{1}{4} & \frac{3}{4} \\
    \frac{1}{3} & \frac{2}{3} & -\frac{3}{4} & \frac{3}{4} \\
    \frac{1}{3} & \frac{2}{3} & \frac{1}{4} & -\frac{1}{4}
    \end{pmatrix}.
    \]
    The source-source reflection block in this case is 
    \[
    S_{\Sigma,\Sigma}
    =
    \begin{pmatrix}
    -\frac{2}{3} & \frac{2}{3} \\
    \frac{1}{3} & -\frac{1}{3}
    \end{pmatrix},
    \]
    which clearly annihilates synchronized incoming data from the source edges. As a result, 
    \[
    S\left[\begin{array}{c}
    a\\
    a\\
    0\\
    0
    \end{array}\right]=\left[\begin{array}{c}
    0\\
    0\\
    a\\
    a
    \end{array}\right],
    \]
    that is, no reflection in the source edges. 
    
    As shown in \cref{fig:deg4}, the two incoming pulses interact with the junction without generating reflected waves on the source edges. The transmitted waves enter the two initially quiescent edges with the same amplitude as the incoming synchronized source mode, consistent with \cref{prop:balance_rank}.
    \begin{figure}
        \centering
        \subfloat[$t=0$]{\label{deg4_t0}
				\includegraphics[height=1.7in,keepaspectratio]{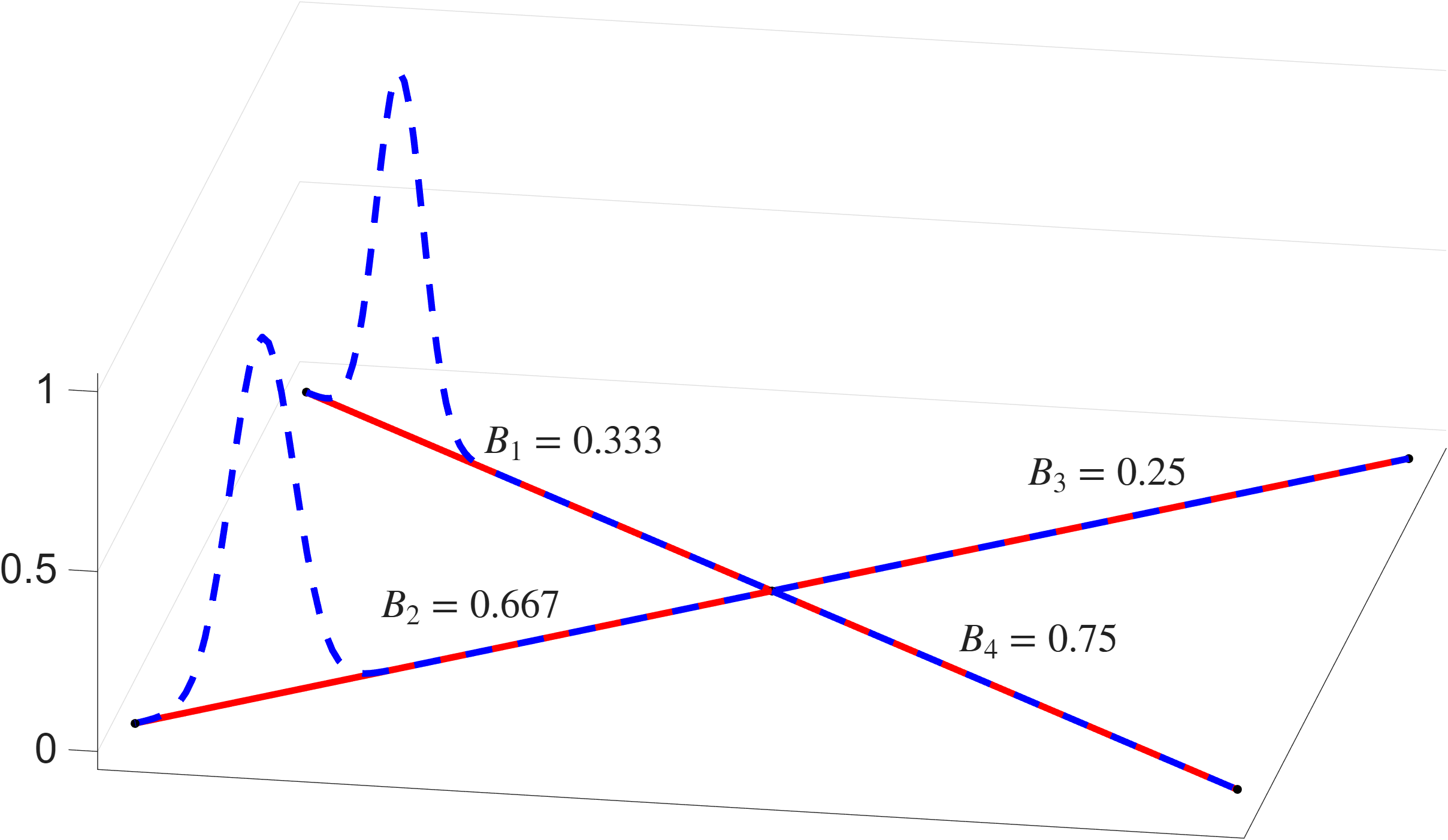}}\\			
        \subfloat[$t=0.8$]{\label{deg4_t0p8}
            \includegraphics[height=1.7in,keepaspectratio]{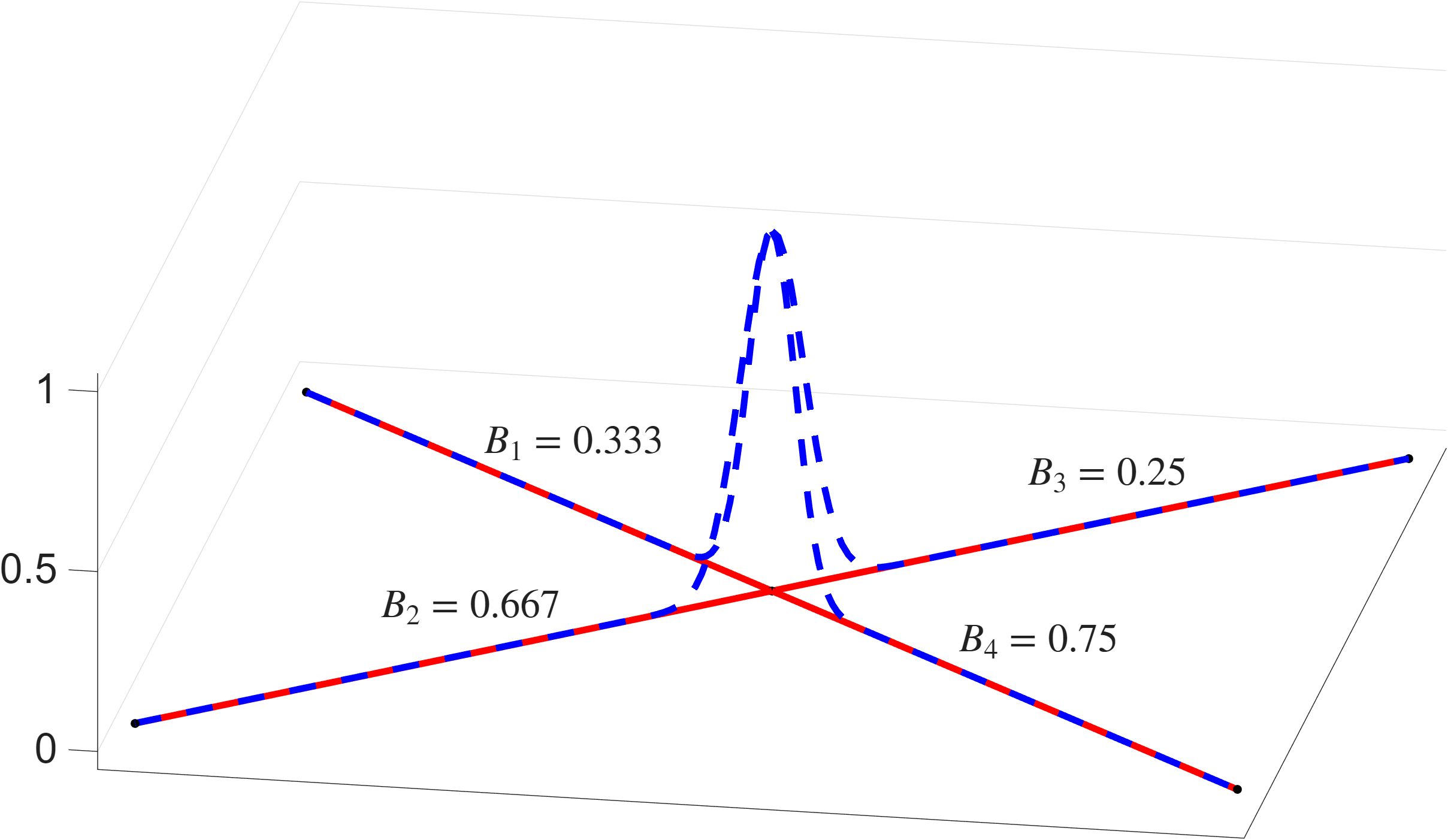}}\\        
        \subfloat[$t=1.6$]{\label{deg4_t1p8}
            \includegraphics[height=1.7in,keepaspectratio]{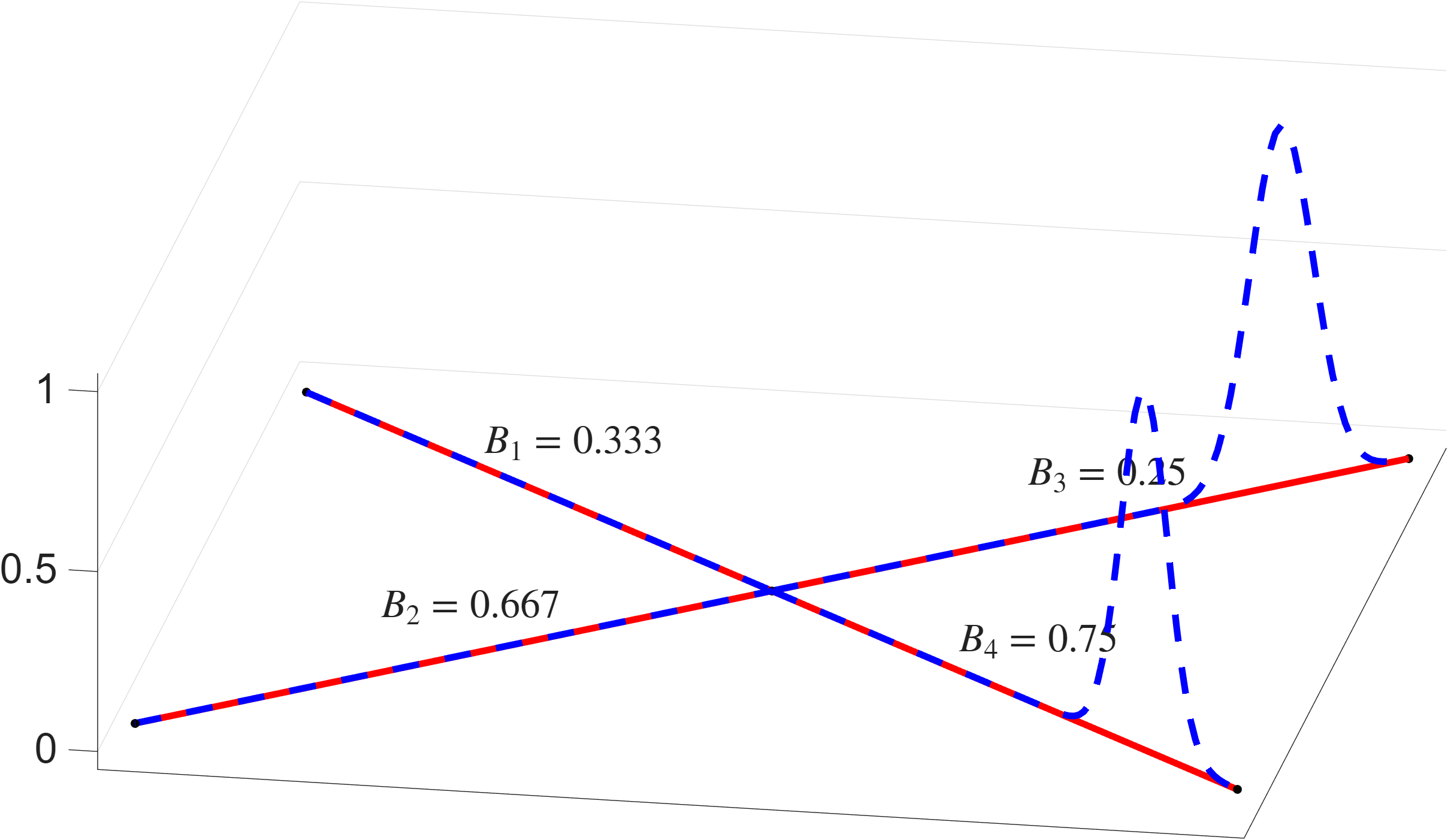}}
        \caption{Numerical visualization of a degree-four star graph with incoming waves on two source edges.  The source edges have widths \(B_1=1/3\) and \(B_2=2/3\), while the initially quiescent edges have widths \(B_3=1/4\) and \(B_4=3/4\).  Since \(B_1+B_2=B_3+B_4\), the vertex is balanced relative to the source set \(\Sigma=\{e_1,e_2\}\).  The snapshots at \(t=0\), \(t=0.8\), and \(t=1.6\) show that synchronized incoming pulses generate no reflection on the source edges and transmit into the complementary edges.}
        \label{fig:deg4}
    \end{figure}

    \subsection{Three-path first-generation reflection}

    For the minimal nontrivial multi-path case, take \(N=3\) for an island with equal internal widths and unequal path lengths.  The incident pulse is again the Gaussian in \cref{IC}; the reflected trace is measured on \(e_0\) at the upstream junction \(\nu_D\).  Since the Gaussian center reaches \(\nu_D\) at time \(t_D\), we plot the trace against the local time \(\tau=t-t_D\).  This aligns the numerical trace with the convention used in \cref{eq:beta0_N_time}, where \(a(t)\) denotes the incoming trace at \(\nu_D\), not the initial profile on \(e_0\).
    
    \begin{figure}[!t]
		\centering
		\includegraphics[width=0.48\textwidth,keepaspectratio]{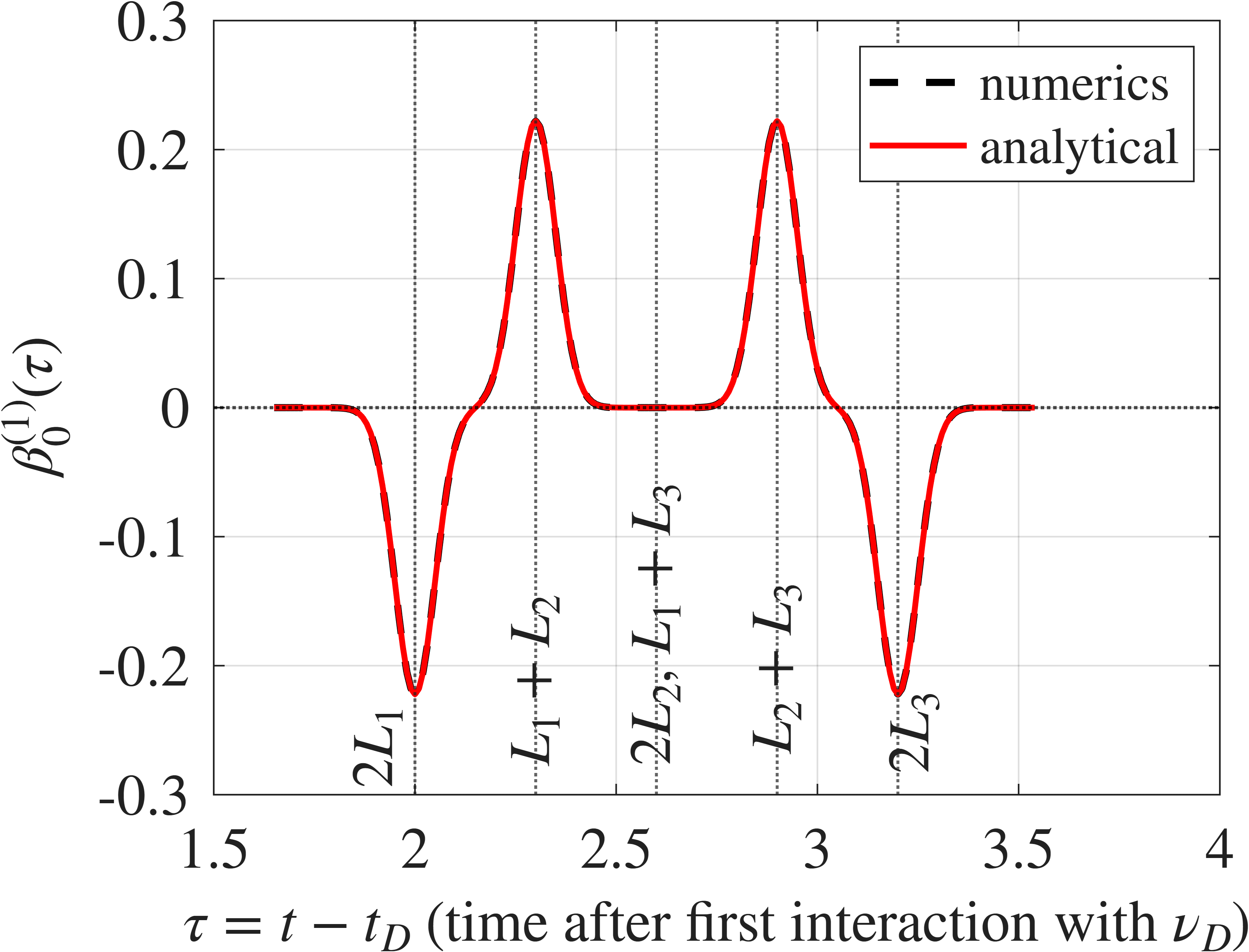}
		\caption{Time-domain verification of \cref{eq:beta0_N_time} for a balanced three-path island.  The reflected characteristic amplitude on \(e_0\) is measured at \(\nu_D\) and plotted against \(\tau=t-t_D\), where \(t_D\) is the arrival time of the incident Gaussian at \(\nu_D\).  The dashed curve is the numerical trace from the Lagrangian characteristic scheme, and the solid curve is the analytical first-generation prediction.  The labels \(\tau=2L_k\) mark same-path return terms, while
        \(\tau=L_k+L_j\) mark mixed-path return terms in \cref{eq:beta0_N_time}.    
		}
		\label{fig:N3_time_domain}
	\end{figure}

    \cref{fig:N3_time_domain} verifies the delayed-sum structure in \cref{eq:beta0_N_time}.  The negative same-path terms \(-a(\tau-2L_k)\) and the positive mixed-path terms \(2a(\tau-L_k-L_j)\) appear at the labeled delay times and combine to produce the observed alternating signs. The agreement between the numerical trace and the analytical curve shows that the formula accounts for all first-generation returns generated by the first interaction at \(\nu_C\), including both same-path and cross-path scattering.
    
	\subsection{Frequency-domain cancellation}\label[appendix]{app:frequency_cancellation}
    
	The frequency-domain form in \cref{eq:beta0_N_freq} separates the incoming signal from the geometry: \(\widehat{\beta}^{(1)}_0(\omega) = \widehat a(\omega)\,\mathcal H(\omega)\), so \(|\mathcal H(\omega)|\) is the first-generation reflection multiplier of the island geometry alone. In \cref{fig:N3_reflection_factor}, we show this factor for three equal-width three-path islands. The equal-length case \(L=\{1,1,1\}\) gives \(\mathcal H \equiv 0\), as predicted by part 1 of \cref{cor:codim_lengths}. For \(L=\{1,1.2,1.4\}\) (\blue{blue} curve in \cref{fig:N3_reflection_factor}), the length differences are integer multiples of $0.2$, so part 2 of \cref{cor:codim_lengths} predicts zeros at \(\omega_m = \frac{2\pi m}{0.2} = 10\pi m\) -- note that $|\mathcal{H}(\omega)|$ is smooth at these roots (because these roots are in fact double zeros due to the sine squared in \cref{eq:sine_squared}). Furthermore, this geometry yields two more roots arising from the uniform length difference $\Delta L=0.2$ and the equal width $\pi_k=1/3$. Indeed, define $q=e^{-i\omega\Delta L}$, we write out $P(\omega)$ for $N=3$
    \[
    \begin{aligned}
    P(\omega)
    &=
    \frac{1}{3}
    \left(
    e^{-iL_1 \omega} + e^{-iL_2 \omega} + e^{-iL_3 \omega}
    \right)\\
    &=
    \frac{e^{-iL_1 \omega}}{3}
    \left(
    1 + e^{-i \Delta L\omega} + e^{-i 2\Delta L\omega}
    \right)\\
    &= \frac{e^{-i\omega}}{3}
    \left( 1 + q + q^2 \right),
    \end{aligned}
    \]
    which vanishes if $q$ is the cube root of unity, i.e. when $\omega_m = \frac{2\pi m}{3\Delta L}$. Similarly,
    \[
    \begin{aligned}
    P(2\omega_m)
    &=
    \frac{e^{-2i\omega_m}}{3}
    \left( 1 + e^{-2i\omega_m\Delta L} + e^{-4i\omega_m\Delta L} \right)\\
    &=
    \frac{e^{-2i\omega_m}}{3} \left( 1+q^2+q^4 \right).
    \end{aligned}
    \]
    Since \(q^4=q\), we also have
    \[
        1+q^2+q^4=1+q^2+q=0,
    \]
    and therefore $P(2\omega_m)=0$. Thus $\mathcal H(\omega_m) = P(\omega_m)^2-P(2\omega_m) = 0$. Note in \cref{fig:N3_reflection_factor} that $|\mathcal{H}(\omega)|$ is nonsmooth at the zeros arising from the cube root of unity (nonsmooth because these zeros are simple, and the modulus produces a corner in $|\mathcal{H}|$).

    However, for the unequal length \(L=\{1,1.2,1.5\}\) where we no longer have a uniform length difference, the roots of unity do not appear. The only roots arise from part 2 of \cref{cor:codim_lengths}. In this case, the length differences are integer multiples of $d=0.1$; indeed, observing in \cref{fig:N3_reflection_factor}, the first nonzero root of $|\mathcal{H}(\omega)|$ is at $\omega_0 = \frac{2\pi}{0.1} = 20\pi$ (incidentally coinciding with a root for the uniform length case). 
    
    Beyond $N=3$, $\mathcal H\left(\omega\right)$ can vanish at frequencies not captured by part 2 of \cref{cor:codim_lengths}, i.e. for length tuples that share no common divisor. Consider $N=4$ with equal width $\pi_k = 1/4$. Then, following the definition of $\mathcal H\left(\omega\right)$ we have
    \[
    \mathcal{H}(\omega) = P(\omega)^2 - P(2\omega)= \left(\frac{1}{4}\sum_{k=1}^{4} \zeta_k\right)^2 - \frac{1}{4} \sum_{k=1}^{4} \zeta_k^2 
    \]
    where $\zeta_k = e^{-i\omega L_k}$. We choose $\{\zeta_k\}_{k=1}^{4} = e^{i\theta_0}\{1,e^{i\theta},1,e^{-i\theta}\}$ for $\theta$ and $\theta_0$ to be determined. Substituting these complex numbers back into the sum and simplifying using the double angle formula for cosine, we have
    \[
    \mathcal{H}(\omega) = \frac{e^{2i\theta_0}}{4}\left(1-\cos\theta\right)\left(1+3\cos\theta\right).
    \]
    It is clear that $\theta^* = 0$ corresponds to the equal length case; meanwhile, $\theta^* = \arccos\left(-\frac{1}{3}\right)$ provides a family of lengths. One can choose $\omega = 2\pi$, then 
    \[
        \left\{ L_{1},L_{2},L_{3},L_{4}\right\} =\left\{ 1,2-\frac{\theta^{*}}{2\pi},2,1+\frac{\theta^{*}}{2\pi}\right\} 
    \]
    is one set of admissible, irrational lengths (among many equivalent ones) and thus incommensurable in the sense of \cref{cor:codim_lengths} part 2, and ultimately kills $\mathcal{H}$. A complete classification of the cancellation variety at $N\ge 4$ is left to future work.

    \begin{figure}[!t]
		\centering
		\includegraphics[width=0.48\textwidth,keepaspectratio]{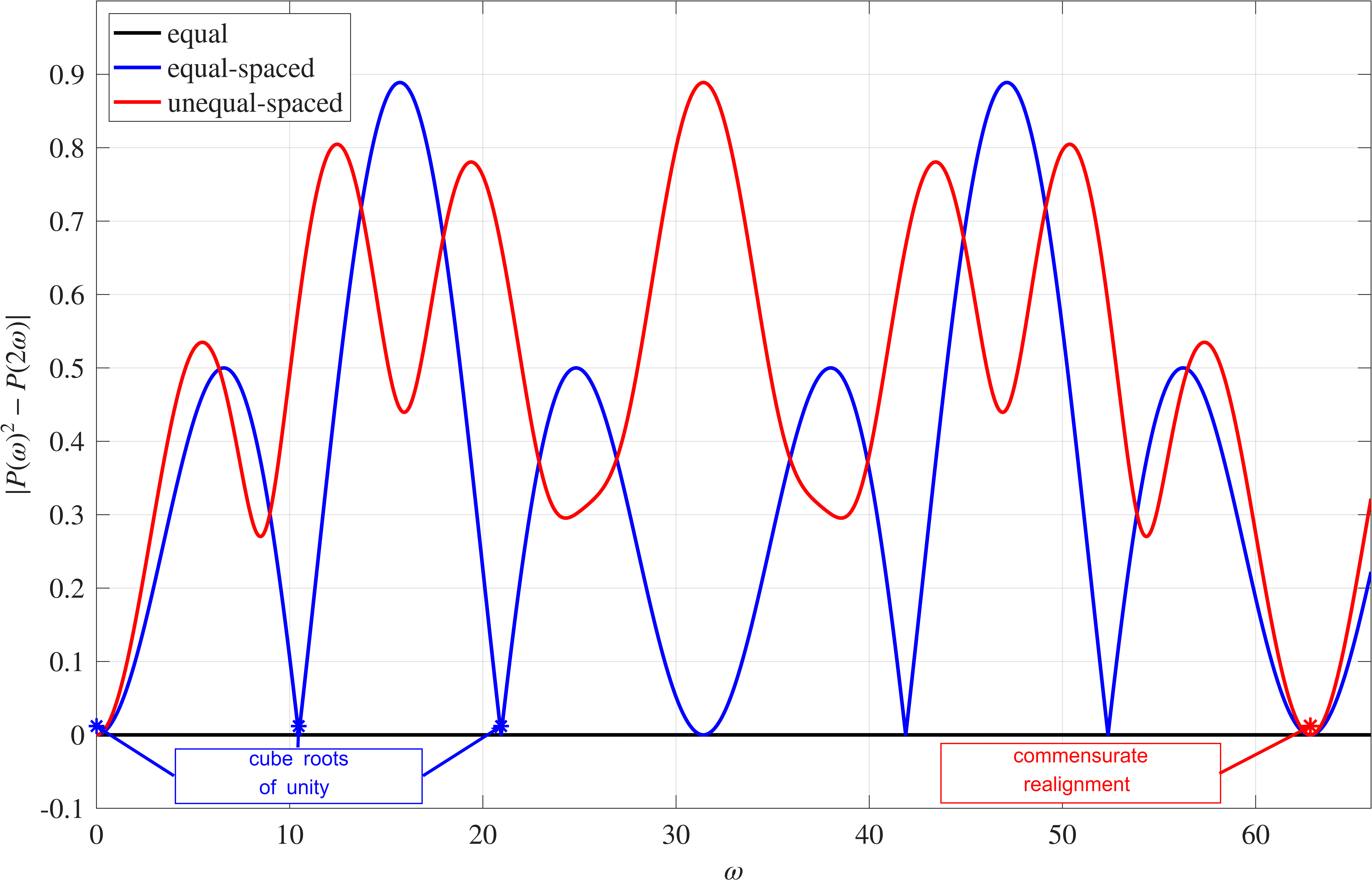}
        \caption{Reflection factor \(\left|P(\omega)^2-P(2\omega)\right|\) for a balanced three-path island at uniform \(\pi_k=1/3\), comparing three length configurations.  Equal lengths \(L=\{1,1,1\}\) (black): \(\mathcal H\equiv 0\), in agreement with part 1 of \cref{cor:codim_lengths}.  Equal-spaced \(L=\{1.0,1.2,1.4\}\) (\blue{blue}): all pairwise differences share \(d=0.2\), so part 2 of \cref{cor:codim_lengths} gives zeros at \(\omega_m = 10\pi m\); the blue curve additionally vanishes at \(\omega=(10\pi/3)m\), where the three complex exponentials \(e^{-i\omega L_k}\) align as cube roots of unity --- a $\pi$-dependent cancellation requiring the uniform \(\pi_k=1/3\) used here.  Unequal-spaced \(L=\{1.0,1.2,1.5\}\) (\textcolor{red}{red}): pairwise differences share \(d=0.1\), so part 2 of \cref{cor:codim_lengths} gives zeros at \(\omega_m = 20\pi m\), the first appearing near \(\omega\approx 62.8\); no cube-roots-of-unity zeros occur as the lengths are not equally spaced.  Both unequal-length curves display the parabolic near-zero growth \(\mathcal H(\omega)=\omega^2\operatorname{Var}_\pi(X)+O(\omega^3)\) predicted by the low-frequency expansion per \cref{eq:variance}.}
		\label{fig:N3_reflection_factor}
	\end{figure}

    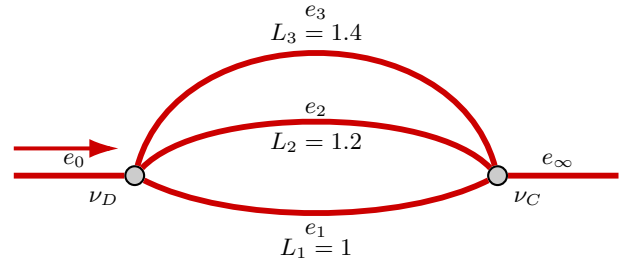
\begin{figure}[!t]
		\centering
		\begin{tikzpicture}[
			scale=0.8,
			junction/.style={circle, draw=black, fill=gray!40, thick, inner sep=2.5pt},
			edge/.style={line width=2.2pt, red!80!black},
			flow/.style={-{Latex[length=4mm,width=2.5mm]}, line width=1.5pt, red!80!black},
			lab/.style={font=\small}
			]
			
			\node[junction, label=below left:{$\nu_D$}] (D) at (0,0) {};
			\node[junction, label=below right:{$\nu_C$}] (C) at (6,0) {};
			
			\draw[edge] (-2,0) -- (D);
			\draw[edge] (C) -- (8,0);
			
			\draw[edge] (D) .. controls (1.6,-0.8) and (4.4,-0.8) .. (C);
			\draw[edge] (D) .. controls (1.1,1.15) and (4.9,1.15) .. (C);
			\draw[edge] (D) .. controls (0.75,2.65) and (5.25,2.65) .. (C);
			
			\draw[flow] (-2,0.45) -- (-0.25,0.45);
			
			\node[lab, above] at (-1,0) {$e_0$};
			\node[lab, above] at (7,0) {$e_\infty$};
			
			\node[lab, below] at (3,-0.65) {$e_1$};
			\node[lab, below] at (3,-0.9) {$L_1=1$};
			
			\node[lab, above] at (3,0.85) {$e_2$};
			\node[lab, below] at (3,0.85) {$L_2=1.2$};
			
			\node[lab, above] at (3,2.45) {$e_3$};
			\node[lab, above] at (3,2.05) {$L_3=1.4$};
			
		\end{tikzpicture}	
		\caption{Three-path island used for the frequency-selective numerical test.  The
			internal arc lengths are \(L_1=1\), \(L_2=1.2\), and \(L_3=1.4\), with equal
			internal widths \(B_1=B_2=B_3\) and boundary widths $B_0 = B_\infty = B_1+B_2+B_3$. Both junctions are then balanced relative to their respective source sets: $\nu_D$ relative to $\Sigma_D = {e_0}$, and $\nu_C$ relative to $\Sigma_C = {e_1,e_2,e_3}$. The incoming disturbance from $e_0$ splits at $\nu_D$, recombines at $\nu_C$, and may generate an upstream reflection signal on $e_0$.}
		\label{fig:3-path-island}
	\end{figure}

    \begin{figure*}[!t]
		\centering
		\includegraphics[width=0.8\textwidth,keepaspectratio]{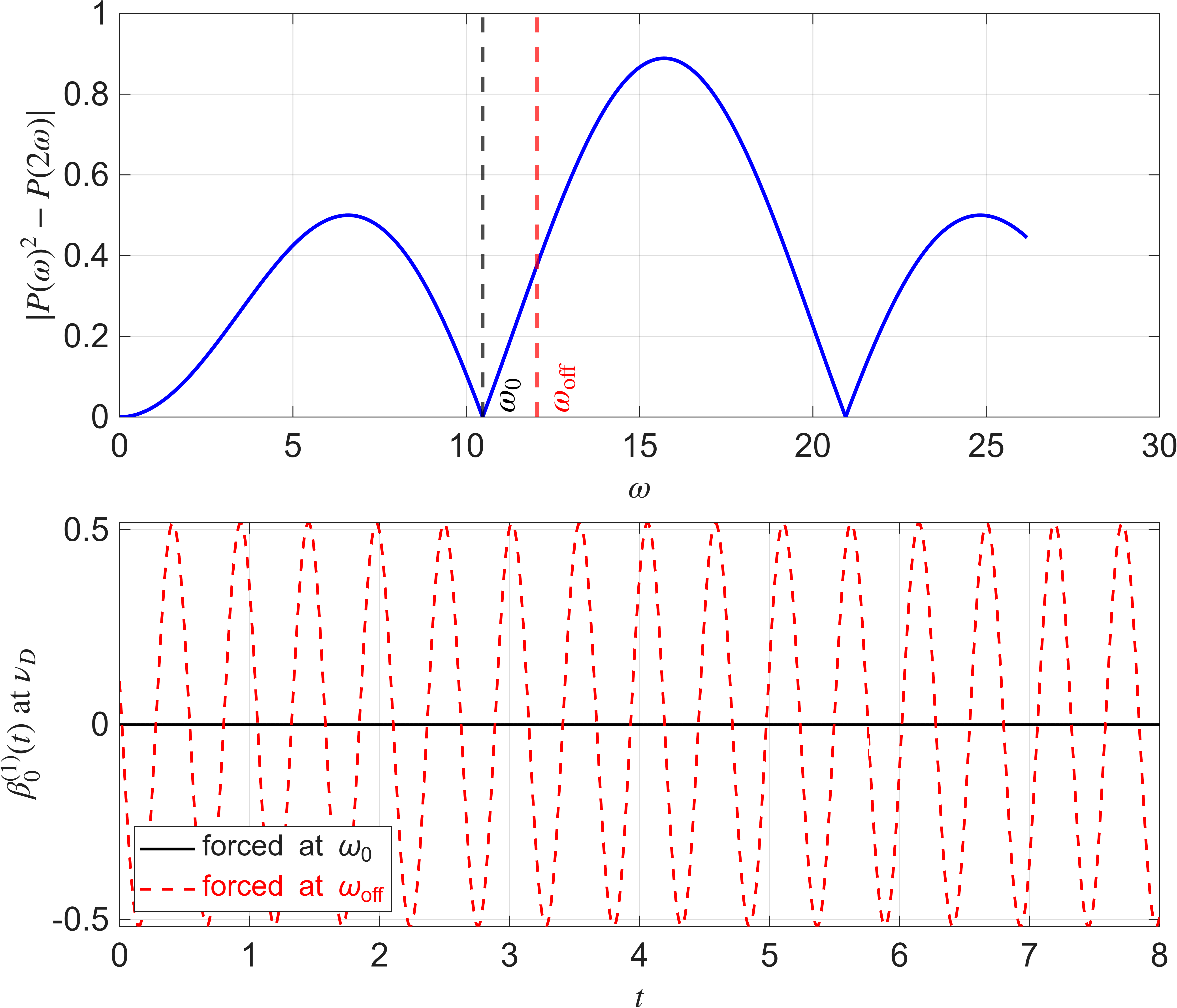}
		\caption{Frequency-selective cancellation for a balanced three-path island. \textbf{Top}: the first-generation reflection factor \(\left|P(\omega)^2-P(2\omega)\right|\). Here \(B_1=B_2=B_3\) and \(L=(1,1.2,1.4)\).  The selected frequency \(\omega_0\) is chosen so that \(\exp(-i\omega_0 L_k)\) form the cube roots of unity up to a common phase, giving \(P(\omega_0)^2-P(2\omega_0)=0\).  The off-null frequency \(\omega_{\rm off}\) is shown for comparison. \textbf{Bottom}: reflected boundary traces $\beta_0^{(1)}(t)$  measured at \(\nu_D\) under purely harmonic boundary forcing on \(e_0\), with every edge pre-initialized at \(t=0\) to the exact steady-state harmonic profile. Forcing at \(\omega_0\) drives \(\beta_0^{(1)}(t)\) to numerical zero for all \(t\), confirming \(\mathcal{H}(\omega_0)=0\); forcing at \(\omega_{\rm off}\) produces a persistent oscillatory reflection with amplitude consistent with the analytical \(|\mathcal{H}(\omega_{\rm off})|\).
        }
		\label{fig:N3_selective_frequency}
	\end{figure*}

    We now run the time-domain test for the equal-spaced configuration (in {\color{blue}blue} in \cref{fig:N3_reflection_factor}). This geometry is shown in \cref{fig:3-path-island}; it is not broadband transparent, but its lengths are chosen so that the first-generation reflection multiplier vanishes at one prescribed frequency. The internal widths are equal and
    \[
        L_1=1,\quad L_2=1+\Delta L,\quad L_3=1+2\Delta L,
        \qquad \Delta L=0.2.
    \]
    Choosing $\omega_0=\frac{2\pi}{3\Delta L}$ as the matching frequency makes the phase factors \(e^{-i\omega_0L_k}\) the cube roots of unity up to a common phase. Consequently, $P(\omega_0)=P(2\omega_0)=0$, and therefore \(\mathcal{H}(\omega_0)=P(\omega_0)^2-P(2\omega_0)=0\).  This pointwise reflection cancellation is frequency-selective.

    For the time-domain test, we prescribe the incoming characteristic data at the upstream boundary of \(e_0\) as a pure harmonic,
    \[
        z_{e_0}(0,t)=\sin(-\omega t),\qquad t\ge 0.
    \]
    At \(t=0\) we initialize every edge with the exact harmonic steady-state solution at frequency \(\omega\). On the internal edges we have both characteristic components \(z\) and \(v\) activated.  
    
    We run this test with two forcing frequencies upstream: the selected matching/reflectionless frequency \(\omega=\omega_0\), where \(\mathcal{H}(\omega_0)=P(\omega_0)^2-P(2\omega_0)=0\), and a nearby (perturbed) frequency \(\omega_{\rm off}=\omega_0+\Delta \omega\), where \(\mathcal{H}(\omega_{\rm off})\neq 0\). The lower panel of \cref{fig:N3_selective_frequency} shows the resulting reflected boundary traces \(\beta_0(t)\) at \(\nu_D\) over the full time window.  Incoming harmonics at \(\omega_0\) (black) holds \(\beta_0^{(1)}(t)= 0\)  for all \(t\), confirming the analytical cancellation \(\mathcal{H}(\omega_0)=0\) directly in the time domain.  
    An upstream incoming harmonics at \(\omega_{\rm off}\) (red) produces a clean steady sinusoid reflection trace throughout, with amplitude in line with the analytical \(|\mathcal{H}(\omega_{\rm off})|\). 

	
	Together, \cref{fig:N3_time_domain,fig:N3_reflection_factor,fig:N3_selective_frequency} numerically illustrate the two forms of \cref{prop:N_path_reflection}.  The time-domain comparison checks the delayed-sum formula in \cref{eq:beta0_N_time}. The reflection-factor plot checks the broadband cancellation criterion, while the boundary-forced test illustrates the frequency-selective cancellation mechanism described in \cref{cor:codim_lengths}.

    \newpage 
    
	\bibliographystyle{plainnat}
	\bibliography{references}
	
\end{document}